\documentclass[a4paper,10pt]{article}
\usepackage[utf8]{inputenc}
\usepackage{amsmath,amssymb,amsthm,amsfonts}
\usepackage{mathrsfs}
\usepackage{graphicx}
\usepackage[caption=false]{subfig}
\usepackage[pdftex,bookmarks=false,colorlinks=true,linkcolor=blue,
citecolor=blue,filecolor=black,urlcolor=blue]{hyperref}

\newtheorem{theorem}{Theorem}
\newtheorem{lemma}[theorem]{Lemma}

\newtheorem{algorithm}[theorem]{Algorithm}

\numberwithin{theorem}{section}
\numberwithin{equation}{section}
\numberwithin{figure}{section}

\title{Numerical methods for a Kohn-Sham density functional model based on optimal transport}
\author{Huajie Chen
\thanks{Zentrum Mathematik, Technische Universit\"{a}t M\"{u}nchen,
Boltzmannstra{\ss}e 3, 85747 Garching, Germany.
E-mail: {\tt chenh@ma.tum.de}.},
Gero Friesecke
\thanks{Zentrum Mathematik, Technische Universit\"{a}t M\"{u}nchen,
Boltzmannstra{\ss}e 3, 85747 Garching, Germany.
E-mail: {\tt gf@ma.tum.de}.},
and Christian B. Mendl
\thanks{Zentrum Mathematik, Technische Universit\"{a}t M\"{u}nchen,
Boltzmannstra{\ss}e 3, 85747 Garching, Germany.
E-mail: {\tt mendl@ma.tum.de}.}
}
\date{\today}

\begin{document}
\maketitle

\begin{abstract}
In this paper, we study numerical discretizations to solve density functional models in the ``strictly correlated electrons'' (SCE) framework. Unlike previous studies our work is not restricted to radially symmetric densities. In the SCE framework, the exchange-correlation functional encodes the effects of the strong correlation regime by minimizing the pairwise Coulomb repulsion, resulting in an optimal transport problem. We give a mathematical derivation of the self-consistent Kohn-Sham-SCE equations, construct an efficient numerical discretization for this type of problem for $N = 2$ electrons, and apply it to the H$_2$ molecule in its dissociating limit. Moreover, we prove that the SCE density functional model is correct for the H$_2$ molecule in its dissociating limit.
\end{abstract}

\section{Introduction}
\label{section-intro}

In the ab-initio quantum mechanical modeling of many-particle systems, Kohn-Sham density functional theory (DFT)
\cite{hohenberg64,kohn65} achieves so far the best compromise between accuracy and computational cost, and has
become the most widely used electronic structure model in molecular simulations and material science.
Within the traditional Kohn-Sham formulation, the ground state energy and electron density of an $N$-electron
system can be obtained by minimizing the energy functional
\begin{equation}\label{ks-energy}
E_{\rm KS}\left[\{\phi_i\}_{i=1}^{N}\right] =
\int_{\mathbb{R}^3} \left( \frac{1}{2}\sum_{i=1}^{N}|\nabla\phi_i({\bf r})|^2
+ v_{\rm ext}({\bf r})\rho({\bf r}) ] \right) d{\bf r} + E_{\rm H}[\rho] + E_{\rm xc}[\rho]
\end{equation}
with respect to orbitals $\{\phi_i\}_{i=1}^{N}$ under the constraint $\int_{\mathbb{R}^3}\phi_i\phi_j = \delta_{ij}$.
Here, $\rho({\bf r})=\sum_{i=1}^{N}|\phi_i({\bf r})|^2$ is the electron density,
$v_{\rm ext}$ is the electrostatic attraction potential generated by the nuclei,
$E_{\rm H}[\rho]=\frac{1}{2}\int_{\mathbb{R}^3}\int_{\mathbb{R}^3}\frac{\rho({\bf r})\rho({\bf r'})}
{|{\bf r}-{\bf r'}|}d{\bf r}d{\bf r'}$ is the Hartree energy that describes electron-electron Coulomb repulsion
energy by a mean field approximation, and $E_{\rm xc}[\rho]$ is the so-called exchange-correlation energy
functional that includes all the many-particle interactions.

The major drawback of DFT is the fact that the exact functional for the exchange-correlation energy is not known.
A basic model is local density approximation (LDA) \cite{kohn65,perdew81},
which is still commonly used in practical calculations.
Improvements of this model give rise to the generalized gradient approximation (GGA)
\cite{langreth80,perdew86,perdew96} and hybrid functionals \cite{becke93,lee88,stephens94}.
Although these models have achieved high accuracy for many chemical and physical systems,
there remain well-known limitations.
For example, in systems with significant static correlation \cite{helgaker00},
LDA, GGA, and also hybrid functionals underestimate the magnitude of the correlation energy.
This becomes particularly problematic for the dissociation of electron pair bonds.
A famous example is the dissociating H$_2$ molecule: the widely employed LDA, GGA,
and even hybrid models fail rather badly at describing the energy curve for dissociating H$_2$.
Many efforts have been made in order to make an appropriate ansatz for the exchange-correlation
functional and tackle this problem (e.g., \cite{fuchs05,gruning03}).
In our view, a principal deficiency of these works is the attempt to describe strong correlation within the framework of mean field approximations.

Alternatively, DFT calculations can also be based on the strongly interacting limit of the Hohenberg-Kohn density functional, denoted ``strictly correlated electrons'' (SCE) DFT \cite{gori09, seidl07}.
This approach considers a reference system with complete correlation between the electrons, and is able to capture key features of strong correlation within the Kohn-Sham framework.
The pioneering work \cite{gori10,malet12,seidl07} has shown that the SCE ansatz can describe certain model systems in the extreme strongly correlated regime with higher accuracy than standard Kohn-Sham DFT.
However, the calculations are presently limited to either one-dimensional or spherically symmetric systems.
To our knowledge, there is no SCE-DFT calculation for dissociating the H$_2$ molecule in $\mathbb{R}^3$.

In the SCE-DFT model, the repulsion energy between strongly interacting electrons is related to
optimal transport theory. Optimal transport was historically studied in \cite{monge81} to model the most economical way of moving soil from one area to another, and was further generalized in \cite{kantorovich1940, kanto42} to the Kantorovich primal and dual formulation. The goal is to transfer masses from
an initial density $\rho_A$ to a target density $\rho_B$ in an optimal way such that the ``cost'' $c(x,y)$ for
transporting mass from $x$ to $y$ is minimized (see \cite{villani09} for a comprehensive treatment).
The Coulomb repulsion energy in the SCE-DFT model can be reformulated as the optimal cost of an optimal transport problem, if we identify the marginals with the electron density divided by number of electrons, i.e., $\rho/N$, and the cost function with the electron-electron Coulomb repulsion
\begin{equation}\label{cost_Coulomb}
c_{\rm ee}({\bf r}_1,\dots,{\bf r}_N) = \sum_{1\leq i<j\leq N}\frac{1}{|{\bf r}_i-{\bf r}_j|}.
\end{equation}
For instance, for a two-electron system within the SCE-DFT framework,
the electron repulsion energy for a given single-particle electron density $\rho$ is
\begin{multline}\label{functional-sce}
V_{\rm ee}^{\rm SCE}[\rho] = \min_{\Psi} \bigg\{
\int_{\mathbb{R}^6}\frac{|\Psi({\bf r}_1,{\bf r}_2)|^2}{|{\bf r}_1-{\bf r}_2|} d{\bf r}_1 d{\bf r}_2, \\
\int_{\mathbb{R}^3}|\Psi({\bf r}_1,{\bf r})|^2d{\bf r}_1 =
\int_{\mathbb{R}^3}|\Psi({\bf r},{\bf r}_2)|^2d{\bf r}_2 = \frac{\rho({\bf r})}{2} \bigg\}. \quad
\end{multline}
Strictly speaking, the set of admissible $|\Psi|^2$'s must be enlarged to probability measures in order to allow strict correlation, which corresponds to concentration of the many-body probability density on a lower dimensional subset \cite{cotar13a}. There are several mathematical investigations of the relations between SCE-DFT and optimal transport problems, see \cite{buttazzo12,cotar13a,cotar13b,friesecke13}, but important open problems remain. To our knowledge, the functional derivative of the SCE functional \eqref{functional-sce} (alias optimal cost functional) with respect to the electron density
(alias marginal measure) is not clear from a mathematical point of view.
However, this result is crucial for deriving the Kohn-Sham equations needed in practical calculations.
Numerical algorithms for optimal transport problems are rather sparse.
Explicit solutions for the co-motion functions are known for one-dimensional and spherically symmetric problems \cite{seidl07}, but cannot be generalized to two- and three-dimensional systems. An alternative route might be the Kantorovich dual formulation of the SCE functional \cite{buttazzo12, mendl13}. In a complementary work \cite{vuckovic2014}, the H$_2$ molecule is studied using an ansatz for the dual potential, and there is a recent simulation of a one-dimensional model H$_2$ molecule using the SCE framework \cite{malet14}.

In this paper, we give a mathematical derivation of the Kohn-Sham equations for optimal transport-based DFT which is rigorous up to physically expected smoothness and continuity assumptions (section~\ref{section-ks}), provide an efficient numerical algorithm for discretising and solving the resulting optimal transport problem for the case of two electrons without restriction to radial symmetry (section~\ref{section-numerical}), and then apply this algorithm to a self-consistent DFT simulation of the H$_2$ molecule in the dissociating limit (section~\ref{section-h2}). Finally, we show both numerically and by a rigorous mathematical argument that the SCE-DFT model is accurate for the H$_2$ molecule in the dissociating limit.

\section{Preliminaries}
\label{section-pre}

Consider a molecular system with $M$ nuclei of charges $\{Z_1,\ldots,Z_M\}$,
located at positions $\{{\bf R}_1,\ldots,{\bf R}_M\}$,
and $N$ electrons in the non-relativistic setting.
The electrostatic potential generated by the nuclei is
\begin{equation*}
v_{\rm ext}({\bf r}) = -\sum_{I=1}^{M}\frac{Z_I}{|{\bf r}-{\bf R}_I|},
\quad {\bf r}\in\mathbb{R}^3.
\end{equation*}
Within the DFT framework \cite{hohenberg64,lieb83}, the ground state density and energy of
the system is obtained by solving the following minimization problem
\begin{equation}\label{min-dft}
E_0 = \min_{\rho} \bigg\{ F_{\rm HK}[\rho]+\int_{\mathbb{R}^3}v_{\rm ext}\rho,
~ \rho\geq 0,~\sqrt{\rho}\in H^1(\mathbb{R}^3),~\int_{\mathbb{R}^3}\rho=N \bigg\},
\end{equation}
where $\rho$ is the electron density and $F_{\rm HK}[\rho]$ is the so-called Hohenberg-Kohn functional \cite{hohenberg64}.
$F_{\rm HK}$ is a universal functional of $\rho$ in the sense that it does not depend on the external potential $v_{\rm ext}$.
Unfortunately, no tractable expression for $F_{\rm HK}$ is known that could be used in numerical simulations.
The standard Kohn-Sham DFT \cite{kohn65} treats the system as $N$ non-interacting electrons,
and approximates $F_{\rm HK}[\rho]$ by a summation of the kinetic energy
\begin{multline}\label{energy-ks}
T_{\rm KS}[\rho] = \inf\bigg\{ \frac{1}{2}\sum_{i=1}^N \int_{\mathbb{R}^3}|\nabla\phi_i({\bf r})|^2 d{\bf r},
~ \phi_i\in H^1(\mathbb{R}^3), \\
~ \sum_{i=1}^N|\phi_i({\bf r})|^2 = \rho({\bf r}),~\int_{\mathbb{R}^3}\phi_i\phi_j = \delta_{ij} \bigg\}, \quad
\end{multline}
the Hartree energy $E_{\rm H}[\rho]$, and an exchange-correction energy $E_{\rm xc}[\rho]$, as shown in Eq.~\eqref{ks-energy}.

Since the standard non-interacting model cannot capture the features that result from strong correlation, it is not able to simulate strongly correlated electron systems, like the H$_2$ molecule in its dissociating limit. In contrast to that, the SCE-DFT model \cite{seidl1999, SPL1999, seidl07} starts from the strongly interacting limit (semi-classical limit) of $F_{\rm HK}$,
and gives rise to the following SCE functional (see \cite{friesecke13} for a mathematical justification)
\begin{equation}\label{energy-sce}
V_{\rm ee}^{\rm SCE}[\rho] = \inf \big\{ V_{\rm ee}[\rho_N],
~\rho_N({\bf r}_1,\dots,{\bf r}_N)\geq 0,~\rho_N~{\rm is~symmetric},~\rho_N\mapsto\rho \big\},
\end{equation}
where
\begin{equation}\label{Vee-rhoN}
V_{\rm ee}[\rho_N] = \int_{\mathbb{R}^{3N}}\sum_{1\leq i<j\leq N}
\frac{\rho_N({\bf r}_1,\dots,{\bf r}_N)}{|{\bf r}_i-{\bf r}_j|} d{\bf r}_1\cdots d{\bf r}_N,
\end{equation}
and $\rho_N\mapsto\rho$ means that $\rho$ is the marginal distribution of $\rho_N$, that is to say
\begin{equation*}
\rho({\bf r}) = N\int_{\mathbb{R}^{3(N-1)}}\rho_N({\bf r},{\bf r}_2,\dots,{\bf r}_N) d{\bf r}_2\cdots d{\bf r}_N.
\end{equation*}
The minimization in Eq.~\eqref{energy-sce} is over all symmetric $N$-point probability measures $\rho_N$ which have the given single-particle density $\rho$ as marginal, and yields the minimum of the electronic Coulomb repulsion energy over all such $\rho_N$. The SCE-DFT model takes $V_{\rm ee}^{\rm SCE}[\rho]$ as the only interaction term, replacing $E_{\rm H}[\rho]+E_{\rm xc}[\rho]$ in standard Kohn-Sham DFT.

The minimization task \eqref{energy-sce} is in fact an optimal transport problem with Coulomb cost \cite{cotar13a,buttazzo12,friesecke13}, which has two alternative formulations:
the Monge formulation and the Kantorovich dual formulation.
For the Monge formulation, one uses the ansatz
\begin{equation}\label{eq-rhoN}
\rho_N({\bf r}_1,\dots,{\bf r}_N) = \frac{\rho({\bf r}_1)}{N}
\delta({\bf r}_2-T_2({\bf r}_1))\cdots\delta({\bf r}_N-T_N({\bf r}))
\end{equation}
with $T_i:\mathbb{R}^3\rightarrow\mathbb{R}^3~(i=2,\dots,N)$ the so-called co-motion functions
(also called optimal transport maps), where we use the convention $T_1({\bf r})={\bf r}$.
The above ansatz already appears on physical grounds, without reference to optimal transport theory, in \cite{seidl07}.
Since the $N$-particle distribution $\rho_N$ in \eqref{eq-rhoN} is zero everywhere except on the set
\begin{equation}\label{set-comotion}
M = \{({\bf r},T_2({\bf r}),\dots,T_N({\bf r})),~{\bf r}\in\mathbb{R}^3\},
\end{equation}
it describes a state where the location of one electron fixes all the other $N-1$ electrons through the co-motion functions $T_i$, $i=2,\dots,N$.
The co-motion functions are implicit functionals of the density, determined by the minimization problem \eqref{energy-sce} and a set of differential equations that ensure the invariance of the density under the coordinate transformation ${\bf r}\mapsto T_i({\bf r})$ \cite{seidl07}, i.e.,
\begin{equation}\label{eq-constraint-comotion}
\rho(T_i({\bf r}))dT_i({\bf r}) = \rho({\bf r})d{\bf r}.
\end{equation}
In terms of these functions, the optimal value of \eqref{energy-sce} reads
\begin{equation}\label{problem-comotion}
V_{\rm ee}^{\rm SCE}[\rho] = \frac{1}{N}\int_{\mathbb{R}^3}
\sum_{1\leq i<j\leq N}\frac{\rho({\bf r})}{|T_i({\bf r})-T_j({\bf r})|}d{\bf r}.
\end{equation}
Note that the ansatz \eqref{eq-rhoN} is not in general symmetric under exchanging particle coordinates, nevertheless,
dropping the symmetrization does not alter the minimum value of \eqref{energy-sce}.

Alternatively, one can start from the so-called Kantorovich dual formulation \cite{kanto42}.
It has been shown in \cite{buttazzo12} that the value of $V^{\rm SCE}_{\rm ee}[\rho]$
is exactly given by the maximum of this Kantorovich dual problem
\begin{equation}\label{kantorovich}
V_{\rm ee}^{\rm SCE}[\rho] = \max\bigg\{\int_{\mathbb{R}^3}u({\bf r})\rho({\bf r})d{\bf r},~
\sum_{i=1}^{N}u({\bf r}_i)\leq\sum_{1\leq i<j\leq N}\frac{1}{|{\bf r}_i-{\bf r}_j|} \bigg\}.
\end{equation}
In what follows, we denote the maximizer of \eqref{kantorovich} by $u_{\rho}$,
which is called the Kantorovich potential.
We assume that $u_{\rho}$ is unique and depends continuously on $\rho$ in the sense that
\begin{equation}\label{u_rho_continuous}
u_{\rho_j}\stackrel{*}{\rightharpoonup}u_{\rho} \quad{\rm if}\quad \rho_j\rightarrow\rho
~{\rm in}~L^1(\mathbb{R}^3).
\end{equation}

For numerical implementations, the Kantorovich dual formulation has high complexity due to the
$3N$-dimensionality of the constraints. In comparison, the Monge formulation amounts to a spectacular dimension reduction, in which the unknowns are $N-1$ maps on $\mathbb{R}^3$ instead of one function $\rho_N$ on $\mathbb{R}^{3N}$. However, for practical purposes it is currently restricted to spherically symmetric densities and one-dimensional systems, for which the constraints \eqref{eq-constraint-comotion} can be solved semi-analytically.
Our purpose is to construct an efficient numerical discretization of the Monge formulation for $N = 2$ electrons which is applicable to non-spherical systems.

\section{Kohn-Sham equations for optimal transport based DFT}
\label{section-ks}

By taking $V_{\rm ee}^{\rm SCE}$ as the only interaction term within the Kohn-Sham DFT framework,
we can obtain the ground state approximations of energy and electron density
by solving the following minimization problem
\begin{equation}\label{min-sce}
E_0 = \inf_{\Phi} \left\{ E_{\rm KS}^{\rm SCE}[\Phi],
~\phi_i\in H^1(\mathbb{R}^3),~\int_{\mathbb{R}^3}\phi_i\phi_j=\delta_{ij}\right\},
\end{equation}
where $\Phi=(\phi_1,\dots,\phi_N)$ denotes the Kohn-Sham orbitals and
\begin{equation}\label{eq-energy-sce}
E_{\rm KS}^{\rm SCE}[\Phi]=\frac{1}{2}\sum_{i=1}^{N} \int_{\mathbb{R}^3}|\nabla\phi_i({\bf r})|^2d{\bf r}
+ \int_{\mathbb{R}^3}v_{\rm ext}({\bf r})\rho_{\Phi}({\bf r})d{\bf r} + V_{\rm ee}^{\rm SCE}[\rho_{\Phi}]
\end{equation}
with $\rho_{\Phi}({\bf r})=\sum_{i=1}^N|\phi_i({\bf r})|^2$.
We shall derive the self-consistent Kohn-Sham equations for \eqref{min-sce} in this section.
The key point is to calculate the functional derivative of the SCE functional
$\delta V_{\rm ee}^{\rm SCE}[\rho] / \delta\rho$ with respect to the single particle density $\rho$,
which is the effective one-body potential coming from the interaction term $V_{\rm ee}^{\rm SCE}[\rho]$. In the derivation below, we make various plausible assumptions on the Kantorovich potential such as uniqueness, continuous dependence on the density, and differentiability at relevant points. We believe these assumptions to be correct except possibly in exceptional situations. A fully rigorous treatment without these assumptions would be desirable, but lies beyond the scope of this paper.

Note that the functional $V_{\rm ee}^{\rm SCE}[\rho]$ is not defined on arbitrary densities,
but only on those with $\int_{\mathbb{R}^3}\rho=N$. Therefore, the definition of the functional
derivative only specifies its integral against perturbations in the corresponding ``tangent space'', that is to say perturbations that have integral zero:
\begin{equation}\label{functional-derivative}
\int_{\mathbb{R}^3}\frac{\delta V_{\rm ee}^{\rm SCE}[\rho]}{\delta\rho}\cdot\tilde{\rho}
= \lim_{\varepsilon\rightarrow0}\frac{V_{\rm ee}^{\rm SCE}[\rho+\varepsilon\tilde{\rho}]
-V_{\rm ee}^{\rm SCE}[\rho]}{\varepsilon}
\quad{\rm for~all}~\tilde{\rho}~{\rm with}~\int_{\mathbb{R}^3}\tilde{\rho}=0.
\end{equation}
The following theorem indicates that the functional derivative is nothing but the Kantorovich potential
$u_{\rho}$ with an additive constant.

\begin{theorem}\label{theo-derivative}
Assume that the maximizer $u_{\rho}$ of \eqref{kantorovich} is unique and depends continuously on
the electron density $\rho$ in the sense of \eqref{u_rho_continuous}.
Then  $v_{\rm SCE}[\rho]=\delta V_{\rm ee}^{\rm SCE}[\rho]/\delta\rho$ is a functional derivative of
$V_{\rm ee}^{\rm SCE}$ at point $\rho$ in the sense of \eqref{functional-derivative} if and only if
\begin{equation}\label{sce-kantorovich}
v_{\rm SCE}[\rho]=u_{\rho}+C \quad\text{for any constant } C.
\end{equation}
\end{theorem}
\begin{proof}
For any given single-particle density $\rho$ with $\int_{\mathbb{R}^3}\rho=N$,
and any perturbation $\tilde{\rho}$ with $\int_{\mathbb{R}^3}\tilde{\rho}=0$, we define
\begin{equation*}
D_{\varepsilon} = \frac{V_{\rm ee}^{\rm SCE}[\rho+\varepsilon\tilde{\rho}]
-V_{\rm ee}^{\rm SCE}[\rho]}{\varepsilon}.
\end{equation*}
For simplicity, we assume $\varepsilon>0$. We have from \eqref{kantorovich} that
\begin{equation}\label{proof-2-a}
D_{\varepsilon} = \frac{\int_{\mathbb{R}^3}u_{\rho+\varepsilon\tilde{\rho}}\,(\rho+\varepsilon\tilde{\rho})
- \int_{\mathbb{R}^3}u_{\rho}\,\rho}{\varepsilon}.
\end{equation}
Using the fact that $u_{\rho+\varepsilon\tilde{\rho}}$ and $u_{\rho}$ are maximizers of \eqref{kantorovich}
with electron density $\rho+\varepsilon\tilde{\rho}$ and $\rho$ respectively, we have
\begin{equation}\label{proof-2-b}
\int_{\mathbb{R}^3}u_{\rho}\,(\rho+\varepsilon\tilde{\rho}) \leq
\int_{\mathbb{R}^3}u_{\rho+\varepsilon\tilde{\rho}}\,(\rho+\varepsilon\tilde{\rho})
\quad{\rm and}\quad
-\int_{\mathbb{R}^3}u_{\rho}\,\rho \leq -\int_{\mathbb{R}^3}u_{\rho+\varepsilon\tilde{\rho}}\,\rho.
\end{equation}
Substituting \eqref{proof-2-b} into \eqref{proof-2-a} gives
\begin{equation}\label{proof-2-c}
\int_{\mathbb{R}^3}u_{\rho}\,\tilde{\rho} \leq D_{\varepsilon}
\leq \int_{\mathbb{R}^3}u_{\rho+\varepsilon\tilde{\rho}}\,\tilde{\rho}.
\end{equation}
Under the uniqueness and continuity assumption \eqref{u_rho_continuous},
the right-hand side of \eqref{proof-2-c} converges to the left-hand side as $\varepsilon\rightarrow 0$.
Hence, for any $\tilde{\rho}$ with $\int_{\mathbb{R}^3}\tilde{\rho}=0$,
$\lim_{\varepsilon\rightarrow 0}D_{\varepsilon}$ exists and equals $\int_{\mathbb{R}^3}u_{\rho}\tilde{\rho}$.
This together with definition \eqref{functional-derivative} leads to $v_{\rm SCE}[\rho]=u_{\rho}$.

Note that the map
$\tilde{\rho}\mapsto\int_{\mathbb{R}^3}\frac{\delta V_{\rm ee}^{\rm SCE}[\rho]}{\delta\rho}\tilde{\rho}$
is unique up to an additive constant since $\int_{\mathbb{R}^3}\tilde{\rho}=0$.
Therefore, the functional derivative viewed as a function can be modified by any additive constant $C$.
This completes the proof.
\end{proof}

The Kantorovich potential $u_{\rho}$ is related to the co-motion functions in the Monge formulation, as noted and justified in \cite{seidl07}. In what follows, we give a more mathematical derivation of this relation, which avoids the interpretation of the effective potential as a Lagrange multiplier and clarifies the relationship between the variational principle \eqref{SGS}, introduced in \cite{seidl07}, and the Kantorovich dual variational principle. Note that the interpretation of the effective potential as a Lagrange multiplier coming from the marginal constraint is heuristically correct, but difficult to make rigorous, the difficulties being related to the notorious ``$v$-representability-problem'', as will be discussed elsewhere.
\begin{theorem}\label{theo-potentials}
Let $\rho_N$ be the minimizer of the optimal transport problem \eqref{energy-sce} with given single-particle density $\rho$.
If $u_{\rho}$ is the Kantorovich potential, i.e., the maximizer of \eqref{kantorovich}, and $u_\rho$ is differentiable, then
\begin{equation}\label{v_rhoN}
\nabla u_{\rho}({\bf r}) = \nabla_{{\bf r}}c_{\rm ee}
({\bf r},{\bf r}_2,\dots,{\bf r}_N)~~\text{on}~~\mathrm{supp}(\rho_N).
\end{equation}
In particular, if $\rho_N$ is of the Monge form \eqref{eq-rhoN}, then
\begin{equation}\label{nabla-v}
\nabla u_{\rho}({\bf r}) = \nabla_{\bf r} c_{\rm ee}({\bf r},{\bf r}_2,\dots,{\bf r}_N)
\mid_{{\bf r}_2=T_2({\bf r}),\dots,{\bf r}_N=T_N({\bf r}_N)}.
\end{equation}
\end{theorem}
\begin{proof}
We first note that $V_{\rm ee}^{\rm SCE}[\rho]$ is convex. To see this, let $\rho$ be a convex combination
$(1-t)\rho_A+t\rho_B$ for some $t\in(0,1)$, and $\rho_N^A$, $\rho_N^B$ be the minimizers of \eqref{energy-sce}
corresponding to the single-particle densities $\rho_A$ and $\rho_B$, we have
\begin{equation}\label{sce-convex}
V_{\rm ee}^{\rm SCE}[\rho]\leq (1-t)V_{\rm ee}[\rho_N^A]+tV_{\rm ee}[\rho_N^B]
=(1-t)V_{\rm ee}^{\rm SCE}[\rho_A]+tV_{\rm ee}^{\rm SCE}[\rho_B].
\end{equation}

Since $V_{\rm ee}^{\rm SCE}[\rho]$ is convex, it equals its double Legendre transform.
Denoting the Legendre transform of a functional $F$ by $F^*$, we have
\begin{equation*}
V_{\rm ee}^{\rm SCE*}[v]=\max_{\rho}\left(\int_{\mathbb{R}^3}v\rho-V_{\rm ee}^{\rm SCE}[\rho]\right).
\end{equation*}
By combining the maximization over $\rho$ and minimization over $\rho_N\mapsto\rho$ in \eqref{energy-sce},
we have
\begin{equation}\label{proof-1-a}
\begin{split}
-V_{\rm ee}^{\rm SCE*}[v] &= \min_{\rho}\left(V_{\rm ee}^{\rm SCE}[\rho]-\int_{\mathbb{R}^3}v\rho\right) \\
&= \min_{\rho}\min_{\rho_N\mapsto\rho}
\left(\int_{\mathbb{R}^{3N}}c_{\rm ee}\rho_N - \int_{\mathbb{R}^3}v\rho\right) \\
&= \min_{\rho_N\mapsto\rho} \int_{\mathbb{R}^{3N}}\rho_N({\bf r}_1,\dots,{\bf r}_N)
\left(c_{\rm ee}({\bf r}_1,\dots,{\bf r}_N)-\sum_{i=1}^N v({\bf r}_i)\right).
\end{split}
\end{equation}
Then the double Legendre transform is
\begin{equation}\label{proof-1-b}
V_{\rm ee}^{\rm SCE**}[\rho]=\max_v\left(\int_{\mathbb{R}^3}v\rho-V_{\rm ee}^{\rm SCE*}[v]\right).
\end{equation}
Combining \eqref{proof-1-a}, \eqref{proof-1-b}, and the fact that $V_{\rm ee}^{\rm SCE}$ equals its double
Legendre transform results in the following variational principle
\begin{equation}\label{SGS}
V_{\rm ee}^{\rm SCE}[\rho] = \max_v\left( \int_{\mathbb{R}^3}v\rho
+ \min_{\rho_N}\int_{\mathbb{R}^{3N}}\big(c_{\rm ee}-\sum_{i=1}^N v({\bf r}_i)\big)\rho_N \right).
\end{equation}
Note that the constraint $\rho_N\mapsto\rho$ has been eliminated in \eqref{SGS}.
For any fixed $v$, let $\mathcal{V}({\bf r}_1,\dots,{\bf r}_N)=\sum_{i=1}^N v({\bf r}_i)$.
The inner variational principle of \eqref{SGS} reads
\begin{equation*}
\min_{\rho_N}\int_{\mathbb{R}^{3N}}(c_{\rm ee}-\mathcal{V})\rho_N.
\end{equation*}
Since $c_{\rm ee}-\mathcal{V}$ is a pure multiplicative operator, it follows (provided $v$ is differentiable) that the support of
any minimizer $\rho_N$ must be contained in the set of absolute minimizers of $c_{\rm ee}-\mathcal{V}$.
Note that on the latter set, $\nabla(c_{\rm ee}-\mathcal{V})=0$. Therefore, we have
\begin{equation}\label{proof-1-c}
\nabla_{{\bf r}_i}(c_{\rm ee}-\mathcal{V}) = \nabla_{{\bf r}_i}c_{\rm ee}({\bf r}_1,\dots,{\bf r}_N)
- \nabla_{{\bf r}_i}v({\bf r}_i) = 0 \quad\text{on }\mathrm{supp}(\rho_N)
\end{equation}
for $i=1,\dots,N$.

According to Lemma~\ref{lemma-outerSGS} in the appendix, if $v_0$ is a maximizer of \eqref{SGS},
then the corresponding minimizer $\rho_N^0$ of the inner optimization of \eqref{SGS}
is exactly the minimizer of the original problem \eqref{energy-sce}.
Therefore, \eqref{proof-1-c} implies that if $\rho_N^0$ is a minimizer of \eqref{energy-sce}, and $v_0$ is differentiable, then
\begin{equation}\label{proof-1-z}
\nabla_{{\bf r}_i}v_0({\bf r}_i) = \nabla_{{\bf r}_i}c_{\rm ee}({\bf r}_1,\dots,{\bf r}_N)
\quad{\rm on~supp}(\rho_N^0),~~i=1,\dots,N.
\end{equation}

In order to obtain \eqref{v_rhoN}, it is now only necessary to show that $u_{\rho}({\bf r})=v_0({\bf r})+\mu$ with some constant $\mu$. Note that the maximum value of \eqref{SGS} is invariant under changing $v$ by an additive constant, because the two integrals involving $v$ cancel. Therefore, the maximization over $v$ in \eqref{SGS} may be restricted to $v$'s with the additional property
\begin{equation}\label{proof-1-d}
\min_{({\bf r}_1,\dots,{\bf r}_N)\in\mathbb{R}^{3N}}\left(
c_{\rm ee}({\bf r}_1,\dots,{\bf r}_N) - \sum_{i=1}^N v({\bf r}_i) \right)= 0.
\end{equation}
For these $v$'s, the minimization over $\rho_N$ in \eqref{SGS} can be carried out explicitly (just place the the support of $\rho_N$ at the global minimizers of $c_{\rm ee}-\mathcal{V}$).
It then follows from \eqref{SGS} that
\begin{equation}\label{proof-1-e}
V_{\rm ee}^{\rm SCE}[\rho] = \max\left\{\int_{\mathbb{R}^3}v\rho,~v~{\rm satisfies}~\eqref{proof-1-d}\right\}.
\end{equation}
Since $\int_{\mathbb{R}^3}(v + C)\rho$ is increasing as a function of the additive constant $C$, condition \eqref{proof-1-d}
can be changed into the inequality
\begin{equation}\label{proof-1-f}
\sum_{i=1}^{N}v({\bf r}_i)\leq\sum_{1\leq i<j\leq N}\frac{1}{|{\bf r}_i-{\bf r}_j|}
\end{equation}
without affecting the maximal value in \eqref{proof-1-e}. This yields the Kantorovich dual form \eqref{kantorovich}.
Therefore, any maximizer of \eqref{SGS} satisfies that $u_{\rho}({\bf r})=v_0({\bf r})+\mu$ with some constant $\mu$,
which together with \eqref{proof-1-z} implies \eqref{v_rhoN}.

If the minimizer of \eqref{energy-sce} is of the Monge form \eqref{eq-rhoN},
then we have ${\rm supp}(\rho_N)\subset M$ with $M$ given by \eqref{set-comotion}.
This implies \eqref{nabla-v} and completes the proof.
\end{proof}

Since the Coulomb cost function $c_{\rm ee}$ is given by \eqref{cost_Coulomb}, Eq.~\eqref{nabla-v} is reduced to the following relation according to Theorem~\ref{theo-potentials} (see also \cite{seidl07})
\begin{equation}\label{eq-NbodyOT}
\nabla u_{\rho}({\bf r}) = -\sum_{i=2}^N\frac{{\bf r}-T_i({\bf r})}{|{\bf r}-T_i({\bf r})|^3}.
\end{equation}
In case $N=2$, there is only one co-motion function $T$, and
\begin{equation}\label{eq-2bodyOT}
\nabla u_{\rho}({\bf r}) = -\frac{{\bf r}-T({\bf r})}{|{\bf r}-T({\bf r})|^3}.
\end{equation}
Note that solving this equation for $T({\bf r})$ gives an instance of the celebrated Gangbo-McCann formula \cite{GangboMcCann1996} for the optimal map in terms of the Kantorovich potential. For the Coulomb cost this formula takes the form \cite{cotar13a}
\begin{equation}\label{eq-GangboMcCannFormula}
T({\bf r}) = {\bf r} + \frac{\nabla u_{\rho}({\bf r})}{|\nabla u_\rho({\bf r})|^{3/2}}.
\end{equation}
However, unlike \eqref{eq-GangboMcCannFormula}, formula \eqref{eq-2bodyOT} generalizes (in the form of \eqref{eq-NbodyOT}) to many-body or multi-marginal problems. Thus formula \eqref{eq-NbodyOT} should be viewed as the correct generalization of the Gangbo-McCann formula to multi-marginal problems. We note that our derivation did not make use of Coulombic features of the cost; the same arguments yield a version of Eq.~\eqref{eq-NbodyOT} for general pair costs of form $c_{\rm ee}({\bf r}_1,\dots,{\bf r}_N) = \sum_{i<j} w({\bf r}_i-{\bf r}_j)$, or Eq.~\eqref{nabla-v} for fully general costs.

Using Theorem \ref{theo-derivative} and \ref{theo-potentials}, we can derive the Kohn-Sham equations corresponding to the SCE energy functional \eqref{eq-energy-sce} with a computable effective potential.
It is the Euler-Lagrange equation corresponding to this minimization problem
(after a unitary transformation to diagonalize the symmetric $N\times N$ matrix of Lagrange multipliers):
find $\lambda_i\in\mathbb{R}, ~\phi_i\in H^1(\mathbb{R}^3)~(i=1,2,\dots,N)$ such that
\begin{equation}\label{ks-sce}
\left\{ \begin{array}{rcl}
\Big(-\frac{1}{2}\Delta + v_{\rm ext} + v_{\rm SCE}[\rho_{\Phi}]\Big)
\phi_i &=& \lambda_i\,\phi_i\quad \text{in}\quad\mathbb{R}^3, \quad i=1,2,\dots, N, \\[1ex]
\displaystyle \int_{\mathbb{R}^3}\phi_i\,\phi_j &=& \delta_{ij}.
\end{array} \right.
\end{equation}
This is a nonlinear eigenvalue problem, where the potential $v_{\rm SCE}[\rho_{\Phi}]$ depends on the electron density $\rho_{\Phi}$ associated with the orbitals $\phi_i$.
A self-consistent field (SCF) iteration algorithm is commonly resorted to for this nonlinear problem.
In each iteration step of the algorithm, a new effective potential is constructed from a trial electron
density and a linear eigenvalue problem is then solved to obtain the low-lying eigenvalues.

We shall comment further on the additive constant in \eqref{sce-kantorovich}. The above equations remain valid when $v_{\rm SCE}[\rho_{\Phi}]$ is modified by an arbitrary additive constant. This yields the same Kohn-Sham orbitals $\phi_i$ and only leads to a corresponding shift of the nonlinear eigenvalues $\lambda_i$. However, as pointed out in \cite{friesecke13}, it is only when $v_{\rm SCE}[\rho_{\Phi}]$ is precisely the Kantorovich potential that the ground state energy can equal the sum of Kohn-Sham eigenvalues, i.e.,
\begin{equation*}
E_0 = \sum_{i=1}^N \lambda_i.
\end{equation*}

In summary, an SCF algorithm for solving the Kohn-Sham equation \eqref{ks-sce} is given by
\begin{algorithm}\label{algorithm-ks_SCE}
SCF iterations for SCE-based Kohn-Sham equations
\begin{enumerate}
  \item Given $\epsilon>0$. Let $k=0$ and $\rho_0$ be an initial electron density.
  \item Calculate the co-motion functions from $\rho_k$
        (by using the numerical methods introduced in the next section).
  \item Calculate the effective potential $v_{\rm SCE}[\rho_k]$ by \eqref{eq-NbodyOT}.
  \item Solve the linear eigenvalue problem
        $$
        \left(-\frac{1}{2}\Delta + v_{\rm ext} + v_{\rm SCE}[\rho_{\Phi}]\right)\phi_i
        = \lambda_i\phi_i \quad i=1,\dots,N
        $$
        for low-lying eigenvalues to obtain new Kohn-Sham orbitals,
        from which a new electron density $\rho_k^{\rm out}$ can be calculated.
  \item If $\|\rho_k-\rho_k^{\rm out}\|<\epsilon$, stop; else,
        generate a new electron density $\rho_{k+1}$ by some charge mixing technique
        and go to 2.
\end{enumerate}
\end{algorithm}

\section{Numerical discretizations of optimal \\ transportation}
\label{section-numerical}

In each iteration of the SCF algorithm for solving \eqref{ks-sce}, one has to construct $v_{\rm SCE}[\tilde{\rho}]$ from a trial electron density $\tilde{\rho}$. According to \eqref{eq-NbodyOT}, this requires the solution of the optimal transport problem \eqref{energy-sce} with a given single-particle density to obtain the co-motion functions $T_i,~i=2,\dots,N$.
For simplicity, we only consider the case $N=2$, where only one co-motion function has to be calculated (which is denoted by $T$ in the following). For systems with more than two electrons, we refer to Section \ref{section-perspective} for a future perspective.

We discretize the computational domain into $n$ finite elements $e_1,\dots,e_n$. (We replace $\mathbb{R}^3$ by a bounded domain so that it can be discretized into a finite number of elements. This is reasonable since the electron density $\rho({\bf r})$ of a confined system decays exponentially fast to zero as $|{\bf r}|\rightarrow\infty$ \cite{hoffmann01}.) Each element is represented by a point ${\bf a}_k$ located at its barycenter and its electron mass
$\rho_k=\int_{e_k}\rho({\bf r})d{\bf r}$. Within this discretization, we can approximate the two-particle density $|\Psi({\bf r}_1,{\bf r}_2)|^2$ by a matrix $X=(x_{kl}) \in \mathbb{R}^{n \times n}$ with $x_{kl}=|\Psi({\bf a}_k,{\bf a}_l)|^2$. (Alternatively, one could identify the entries with the average $x_{kl} = \frac{1}{|e_k| \cdot |e_l|}\int_{e_k}\int_{e_l} |\Psi({\bf r}_1,{\bf r}_2)|^2 d{\bf r}_1 d{\bf r}_2$.)
%
The continuous problem \eqref{energy-sce} is then discretized into
\begin{equation}\label{ot_linprog}
\begin{array}{rl}
\displaystyle \min_{X} & \displaystyle \sum_{1\leq k,l\leq n} \frac{x_{kl}}{|{\bf a}_k-{\bf a}_l|} \\ \\
{\rm s.t.} & \sum_{1\leq k\leq n}x_{kl} = \frac{1}{2}\rho_k, \quad l=1,\dots,n \\[1ex]
& \sum_{1\leq l\leq n}x_{kl}=\frac{1}{2}\rho_l, \quad k=1,\dots,n \\[1ex]
& x_{kl}\geq 0. \end{array}
\end{equation}
Note that \eqref{ot_linprog} is a linear programming problem of the form
\begin{equation*}
\begin{split}
&\min_x f^T x \\
&{\rm s.t.}\quad A x = b ~{\rm and}~ x_k \ge 0,
\end{split}
\end{equation*}
where $x$ is the vector containing the entries of $X$. We can solve this problem by standard optimization routines like `\emph{linprog}' in \textsf{Matlab}. Due to the symmetry of the problem, one can assume that $x_{lk} = x_{kl}$ and only needs to consider $x_{kl}$ for $k \le l$.

As a remark, the dual problem of \eqref{ot_linprog} (in the sense of linear programming) results in a discretized version of the Kantorovich dual formulation \eqref{kantorovich}.

The solution of \eqref{ot_linprog} entails an approximation of the co-motion functions at the barycenters $\{{\bf a}_k\}_{1\leq k\leq n}$ via the matrix $X=(x_{kl})$:
\begin{equation}\label{T_approximate}
T_n({\bf a}_k)=\sum_{l=1}^n {\bf a}_l\frac{2 x_{kl}}{\rho_l},
\quad k=1,\dots,n,
\end{equation}
where $x_{kl}$ can also be regarded as the mass of electron transported from ${\bf a}_k$ to ${\bf a}_l$.
If the discretization is sufficiently fine, i.e., $n$ large enough, then $T_n$ is a good
approximation of $T$ (see the following numerical example).

For a \emph{uniform} discretization $\{e_k\}_{1\leq k\leq n}$,
the degrees of freedom for linear programming \eqref{ot_linprog} may be huge.
To reduce the computational cost, we use a locally refined mesh instead, which has more elements where the electron
density is high and less elements where the electron density is low.
Generally speaking, the optimal mesh may be such that each element $e_k$ has almost equal electron mass $\rho_k$.
This type of mesh can be generated by an adaptive procedure, say, one refines the element when its electron mass is
larger than a given threshold and coarse it otherwise.
Since the electron density decays exponentially fast to zero as $|{\bf r}|\rightarrow\infty$, the mesh is much coarser
far away from the nuclei than close to the nuclei, which reduces the degrees of freedom significantly.

As a remark, let us assume for a moment that all elements have exactly the same mass, $\rho_k = \bar{\rho}$ for all $k$. Then the constraints in \eqref{ot_linprog} force $X$ to be a \emph{doubly stochastic} matrix (up to a global scaling factor) with nonnegative entries. According to Birkhoff's theorem, the extremal points of the convex set of admissible matrices $X$ are the permutations, i.e., matrices with exactly one non-zero entry $\frac{1}{2} \bar{\rho}$ in each row (or column). Since the optimum is obtained at an extremal point, the optimizer can be chosen of this form. We have thus derived a discrete analogue of the Monge formulation, since the sum on the right of Eq.~\eqref{T_approximate} will have exactly one nonzero term.

Another important technique to reduce the computational cost is to exploit the symmetry of the system.
If the electron density $\rho$ has some kind of symmetric property, then we can reduce the computations
to some subdomain accordingly.
For example, \cite{seidl07} gives an explicit formula of co-motion functions for spherically symmetric
electron densities by making use of the symmetry. More precisely, it is proven that if the density has the form
$\rho({\bf r})=h(|{\bf r}|)$ with some function $h: [0,\infty) \rightarrow \mathbb{R}$,
then the corresponding co-motion function $T$ has to be spherically symmetric itself, that is
\begin{equation*}
T({\bf r})=g(|{\bf r}|)\frac{\bf r}{|\bf r|},\quad\forall~{\bf r}\in\mathbb{R}^3
\end{equation*}
with some function $g: [0,\infty) \rightarrow \mathbb{R}$.
This reduces the three-dimensional spherically symmetric problem into a one-dimensional problem.

Here we consider cylindrically symmetric systems, for instances, biatomic molecules.
The following theorem states that the co-motion function $T$ inherits the cylindrical symmetriy of the density.

\begin{theorem}\label{theo-symmetric}
Let $N=2$ and denote the cylindrical coordinates by $(\gamma,\varphi,z)$.
If $\rho({\bf r})=\varrho(\gamma,z)$ with some function $\varrho:[0,\infty)\times\mathbb{R}\rightarrow\mathbb{R}$,
then the corresponding co-motion function $T$ satisfies
\begin{equation}\label{T_symmetric}
T:~(\gamma,\varphi,z)\mapsto(\gamma',\varphi+\pi,z') \quad\forall~(\gamma,z)\in [0,\infty)\times\mathbb{R},
\end{equation}
where $(\gamma',z')=\ell(\gamma,z)$ with some map
$\ell:[0,\infty)\times\mathbb{R}\rightarrow [0,\infty)\times\mathbb{R}$.
\end{theorem}
\begin{proof}
Let $\mathcal{R}_{\theta}$ be the rotation operator with angle $\theta$ around the z-axis, i.e., $\mathcal{R}_{\theta}(\gamma,\varphi,z)=(\gamma,\varphi+\theta,z)$.
Let $\rho_2$ be a minimizer of \eqref{energy-sce} with single-particle electron density $\rho$,
such that $\rho_2({\bf r}_1,{\bf r}_2)=\frac{\rho({\bf r}_1)}{2}\delta({\bf r}_2-T({\bf r}_1))$
with $T$ the corresponding co-motion function.

Let $\tilde{\rho}_2({\bf r}_1,{\bf r}_2) = \rho_2(\mathcal{R}_{\theta}{\bf r}_1,\mathcal{R}_{\theta}{\bf r}_2)$.
We claim that $\tilde{\rho}_2$ is also a minimizer of \eqref{energy-sce}.
To see this, we observe that
\begin{equation*}
\tilde{\rho}_2({\bf r}_1,{\bf r}_2) = \frac{\rho(\mathcal{R}_{\theta}{\bf r}_1)}{2}
\delta(\mathcal{R}_{\theta}{\bf r}_2-T(\mathcal{R}_{\theta}{\bf r}_1)),
\end{equation*}
which satisfies the marginal constraint
\begin{equation}\label{proof-6-a}
\tilde{\rho}_2\mapsto\rho
\end{equation}
since $\rho({\bf r})=\rho(\mathcal{R}_{\theta}{\bf r})$.
Moreover, we have
\begin{equation}\label{proof-6-b}
\int_{\mathbb{R}^6}\frac{\tilde{\rho}_2({\bf r}_1,{\bf r}_2)}{|{\bf r}_1-{\bf r}_2|}d{\bf r}_1d{\bf r}_2
= \int_{\mathbb{R}^6}\frac{\rho_2({\bf r}_1,{\bf r}_2)}{|{\bf r}_1-{\bf r}_2|}d{\bf r}_1d{\bf r}_2
\end{equation}
from the fact that the cost function $\frac{1}{|{\bf r}_1-{\bf r}_2|}$ is invariant under the map
$({\bf r}_1,{\bf r}_2)\rightarrow(\mathcal{R}_{\theta}{\bf r}_1,\mathcal{R}_{\theta}{\bf r}_2)$.
\eqref{proof-6-a} and \eqref{proof-6-b} together imply that $\tilde{\rho}_2$ is a minimizer of \eqref{energy-sce}.

Therefore, $\tilde{T}:{\bf r}\mapsto\mathcal{R}_{\theta}^{-1}T(\mathcal{R}_{\theta}{\bf r})$
is also a co-motion function of this problem.
Because the co-motion function is unique in the case of two particles (see \cite{cotar13a}), we have
\begin{equation}\label{proof-3-a}
T(\mathcal{R}_{\theta}{{\bf r}})=\mathcal{R}_{\theta}T({\bf r}).
\end{equation}

Since we minimize
\begin{equation*}
\begin{split}
\int_{\mathbb{R}^3}\frac{\rho({\bf r})}{|{\bf r}-T({\bf r})|}d{\bf r}
&= \int_0^{\infty}\gamma d\gamma\int_{-\infty}^{\infty}dz \int_0^{2\pi}d\varphi
\frac{\rho((\gamma,\varphi,z))}{|(\gamma,\varphi,z)-T((\gamma,\varphi,z))|}d\varphi \\
&\stackrel{{\rm fix}~\varphi}{=} \int_0^{\infty}\gamma  d\gamma\int_{-\infty}^{\infty}dz
\int_0^{2\pi}d\theta \frac{\rho(\mathcal{R}_{\theta}(\gamma,\varphi,z))}{|\mathcal{R}_{\theta}
(\gamma,\varphi,z) - T(\mathcal{R}_{\theta}(\gamma,\varphi,z))|}d\theta \\
&\stackrel{\eqref{proof-3-a}}{=} \int_0^{\infty}\gamma d\gamma\int_{-\infty}^{\infty}dz
\int_0^{2\pi}d\theta \frac{\rho(\mathcal{R}_{\theta}(\gamma,\varphi,z))}{|\mathcal{R}_{\theta}
(\gamma,\varphi,z) - \mathcal{R}_{\theta}T((\gamma,\varphi,z))|}d\theta \\
&= 2\pi \int_0^{\infty} \int_{-\infty}^{\infty}
\frac{\gamma\varrho((\gamma,z))}{|(\gamma,\varphi,z) - T((\gamma,\varphi,z))|}d\gamma dz,
\end{split}
\end{equation*}
$T$ has to satisfy \eqref{T_symmetric} to maximize the denominator.
This completes the proof.
\end{proof}

{\bf Example.} We take a one-dimensional two-electron system as an example to illustrate our numerical method.
Although this type of algorithm is not particularly interesting for one-dimensional systems
since analytical formulations are known \cite{malet14,friesecke13},
it is more suitable to present the numerical results and explain our idea.

Let $\Omega=[-5,5]$ and $\rho(x)=0.4-0.08|x|$, see Figure \ref{fig-rho1d}.
The exact co-motion function can be calculated explicitly
\begin{equation}\label{exact1d_comotion}
T(x)=\begin{cases}
\hspace{8pt} 5\left(1-\big[1-0.5(x+5)(0.4+0.08x)\big]^{1/2}\right) & {\rm if}~x\leq 0, \\
-5\left(1-\big[1-0.5(-x+5)(0.4-0.08x)\big]^{1/2}\right) & {\rm if}~x>0.
\end{cases}
\end{equation}
We observe in Figure \ref{fig-comotion} that the numerical approximations of the co-motion function
can be very accurate (compared with the exact formula \eqref{exact1d_comotion}).
Figure~\ref{fig-convergence} shows the convergence of the uniform numerical approximations.
%
The nonuniform mesh (red elements in Figure~\ref{fig-rho1d}) achieves a much higher accuracy with the same degrees of freedom, see Figure~\ref{fig-adaptive}.
Actually, we observe that the errors of a nonuniform mesh with $n=20$ is even smaller compared to a uniform mesh with $n=40$ on average.

\begin{figure}[!ht]
\centering
\subfloat[electron density]{
\label{fig-rho1d}
\includegraphics[width=0.45\textwidth]{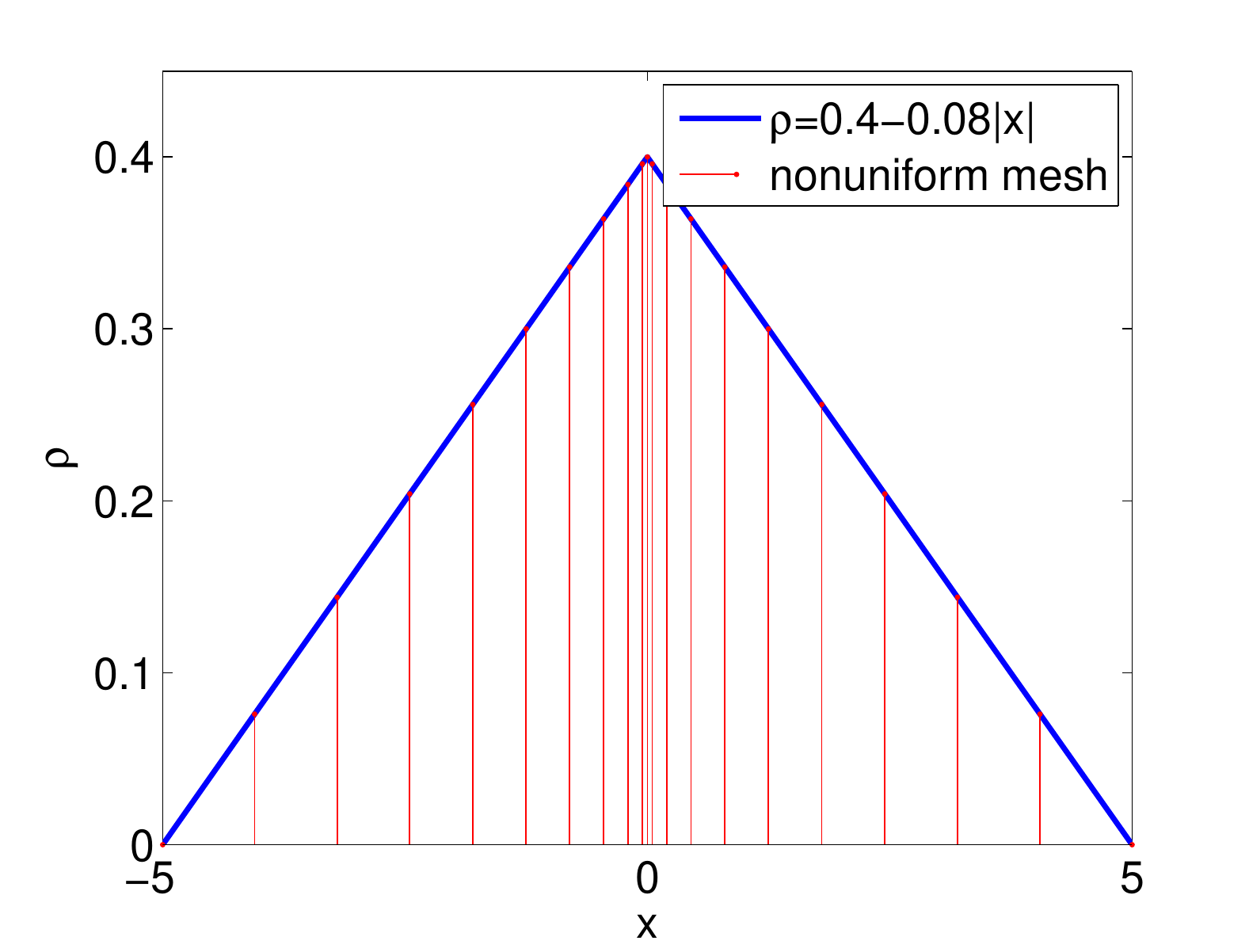}} \quad
\subfloat[co-motion function]{
\label{fig-comotion}
\includegraphics[width=0.45\textwidth]{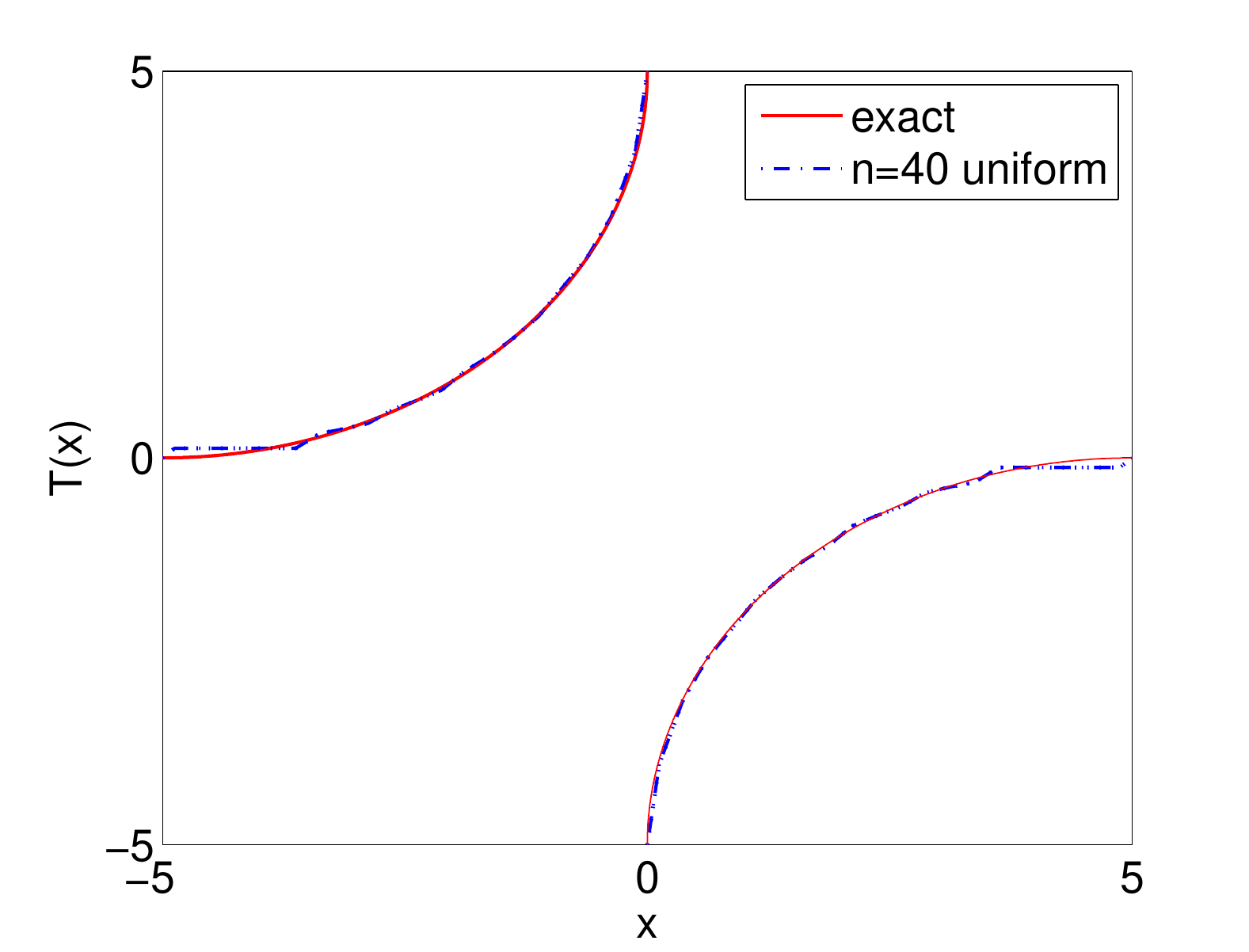}} \\
\subfloat[$T_n$ error, uniform mesh]{
\label{fig-convergence}
\includegraphics[width=0.45\textwidth]{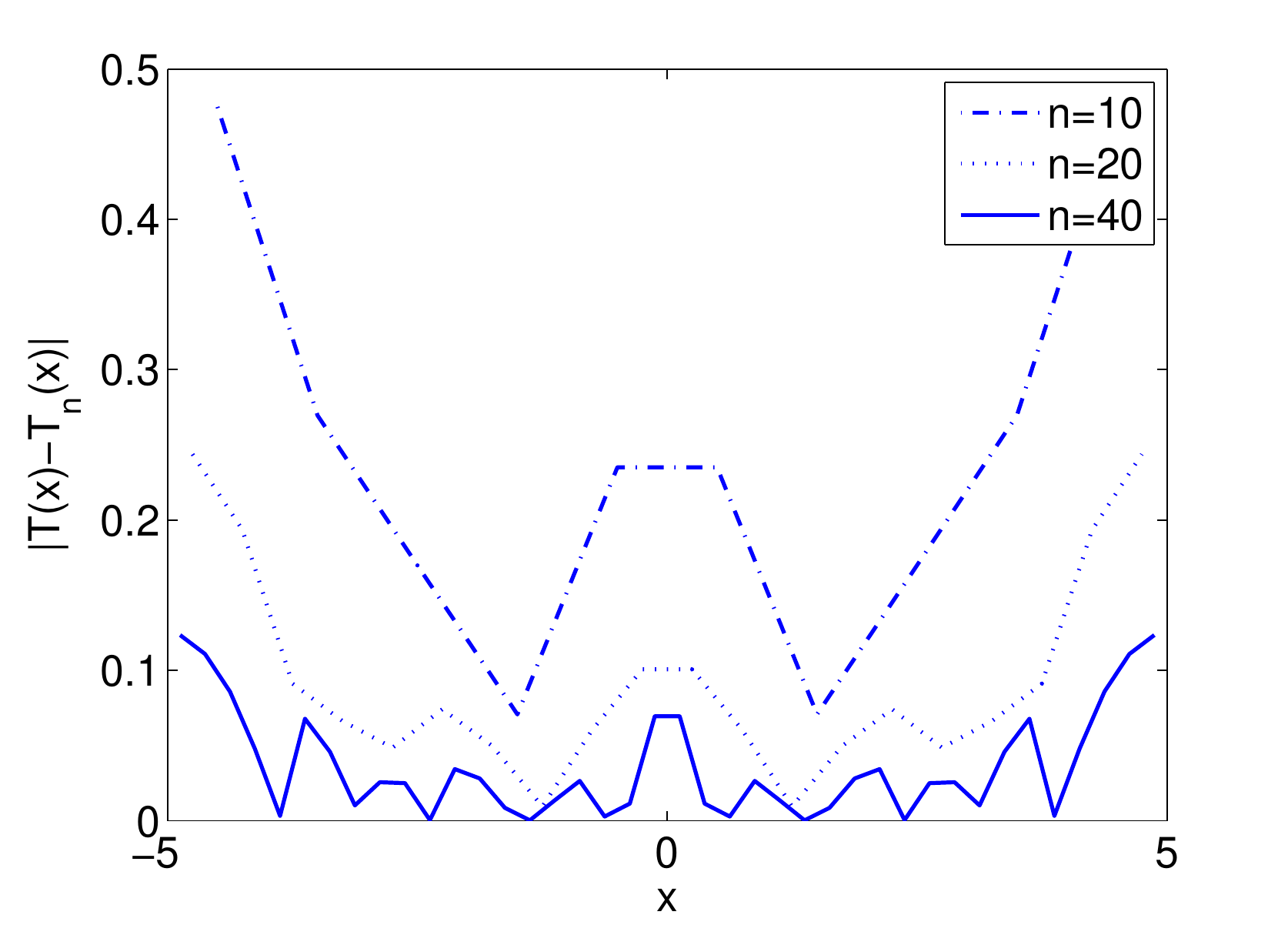}} \quad
\subfloat[$T_n$ error, nonuniform mesh]{
\label{fig-adaptive}
\includegraphics[width=0.45\textwidth]{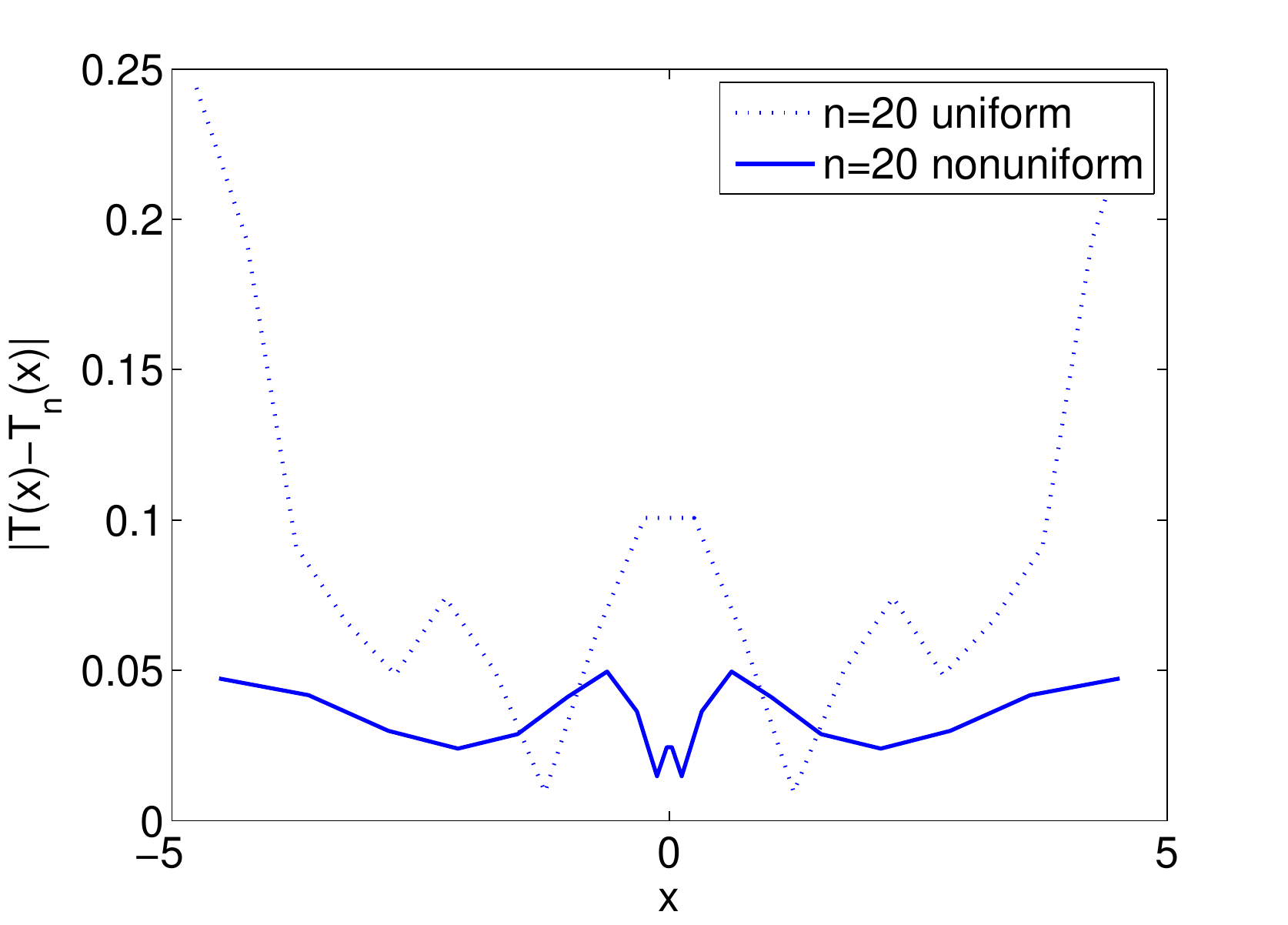}}
\caption{(a) The one-dimensional electron density $\rho$ and a nonuniform discretization. 
(b) The co-motion function corresponding to $\rho$ and its approximation. 
(c) Numerical errors of $T_n$ using uniform meshes, and (d) a nonuniform mesh.}
\end{figure}

Remember that the aforementioned electron density is symmetric in the sense of
\begin{equation*}
\rho(x)=\rho(-x).
\end{equation*}
Therefore, the corresponding co-motion function is symmetric itself, i.e.
\begin{equation*}
T(x) = -T(-x),
\end{equation*}
which maps $[-5,0]$ to $[0,5]$ and maps $[0,5]$ to $[-5,0]$.
Hence, it is only necessary to calculate $T(x)$ on half of the domain, say, $[-5,0]$.

An important application of the numerical methods introduced above is to simulate the H$_2$
molecule at its dissociating limit. We provide more details in the next section.

\section{H$_2$ bond disassociation}
\label{section-h2}

We consider H$_2$ molecule in this section.
Let $R>0$ and ${\bf R}_A=(-R,0,0),~{\bf R}_B=(R,0,0)$ be the locations of two hydrogen atoms. Physically, the hydrogen molecule should dissociate into two free hydrogen atoms as the bond length $2R\rightarrow\infty$, with the ground state spin-unpolarized.
The spin-restricted Hartree-Fock and Kohn-Sham DFT models give the correct spin multiplicity,
but overestimate total energies, i.e., higher than that of two free hydrogen atoms.
In comparison, the spin-unrestricted models give fairly good total energies, while the wave
functions are spin-contaminated, which is known as ``symmetry breaking'' in H$_2$ bond dissociation.

Here we focus on the SCE-DFT model without symmetry breaking, and
show both theoretically and numerically that the restricted Kohn-Sham model \eqref{min-sce}
gives the correct ground state energy in the dissociation limit $R\rightarrow\infty$.
Denote by $e_0$ the ground state energy of a single hydrogen atom
\begin{equation}\label{e_0}
e_0=\inf\left\{\frac{1}{2}\int_{\mathbb{R}^3}|\nabla\phi({\bf r})|^2d{\bf r}
-\int_{\mathbb{R}^3}\frac{|\phi({\bf r})|^2}{|\bf r|}d{\bf r},
~\phi\in H^1(\mathbb{R}^3),~\|\phi\|_{L^2(\mathbb{R}^3)}=1\right\}.
\end{equation}
$E_{\rm SCE}(R)$ denotes the ground state energy of the hydrogen molecule in the SCE-DFT model \eqref{eq-energy-sce}
\begin{multline}\label{E_R}
E_{\rm SCE}(R) = \frac{1}{2R} + \inf \Big\{\int_{\mathbb{R}^3}|\nabla\phi({\bf r})|^2d{\bf r}
+ 2\int_{\mathbb{R}^3}v_{\rm ext}({\bf r})|\phi({\bf r})|^2d{\bf r}
+ V_{\rm ee}^{\rm SCE}[2|\phi|^2], \\
\phi\in H^1(\mathbb{R}^3),~\|\phi\|_{L^2(\mathbb{R}^3)}=1\Big\}, \quad
\end{multline}
where $v_{\rm ext}({\bf r})=-\frac{1}{|{\bf r}-{\bf R}_A|}-\frac{1}{|{\bf r}-{\bf R}_B|}$.
The following result indicates that the SCE-DFT model is correct for the H$_2$ molecule at its dissociating limit.

\begin{theorem}\label{theo-sce-h2}
Let $e_0$ and $E_{\rm SCE}(R)$ be given by \eqref{e_0} and \eqref{E_R} respectively.
We have
\begin{equation}\label{sce-h2-limit}
\lim_{R\rightarrow\infty}E_{\rm SCE}(R) = 2e_0.
\end{equation}
\end{theorem}
\begin{proof}
First, we establish an upper bound of $E_{\rm SCE}(R)$.
Let $\varphi(r)=e^{-r}/\sqrt{\pi}$ and
\begin{equation*}
\psi({\bf r}) = \left(\frac{1}{2} \Big(\varphi^2(|{\bf r}-{\bf R}_A|) + \varphi^2(|{\bf r}-{\bf R}_B|) \Big) \right)^{1/2}.
\end{equation*}
Note that $\varphi$ is the minimizer of \eqref{e_0},
and $\|\varphi\|_{L^2(\mathbb{R}^3)}=1$ implies $\|\psi\|_{L^2(\mathbb{R}^3)}=1$.
We have
\begin{equation}\label{proof-4-b}
E_{\rm SCE}(R) \leq \frac{1}{2R} + \int_{\mathbb{R}^3}|\nabla\psi({\bf r})|^2d{\bf r}
+ 2\int_{\mathbb{R}^3}v_{\rm ext}({\bf r})|\psi({\bf r})|^2d{\bf r} + V_{\rm ee}^{\rm SCE}[2|\psi|^2].
\end{equation}
Let $\phi_1({\bf r})=\varphi(|{\bf r}-{\bf R}_A|)$ and $\phi_2({\bf r})=\varphi(|{\bf r}-{\bf R}_B|)$.
A direct calculation leads to
\begin{equation}\label{proof-4-c}
\begin{split}
\int_{\mathbb{R}^3}|\nabla\psi|^2
&= \int_{\mathbb{R}^3}\frac{|\phi_1\nabla\phi_1+\phi_2\nabla\phi_2|^2}{4(\phi_1^2+\phi_2^2)} \\
&\leq \frac{1}{2} \int_{\mathbb{R}^3}(|\nabla\phi_1|^2+|\nabla\phi_2|^2) = \int_{\mathbb{R}^3}|\nabla\varphi|^2
\end{split}
\end{equation}
and
\begin{equation}\label{proof-4-d}
\begin{split}
& 2\int_{\mathbb{R}^3}v_{\rm ext}({\bf r})|\psi({\bf r})|^2d{\bf r}
= -\int_{\mathbb{R}^3}\frac{\phi_1^2({\bf r})+\phi_2^2({\bf r})}{|{\bf r}-{\bf R}_A|}d{\bf r}
-\int_{\mathbb{R}^3}\frac{\phi_1^2({\bf r})+\phi_2^2({\bf r})}{|{\bf r}-{\bf R}_B|}d{\bf r} \\
&= -2\int_{\mathbb{R}^3}\frac{\varphi^2(|{\bf r}|)}{|{\bf r}|}d{\bf r}
-\int_{\mathbb{R}^3}\frac{\phi_2^2({\bf r})}{|{\bf r}-{\bf R}_A|}d{\bf r}
-\int_{\mathbb{R}^3}\frac{\phi_1^2({\bf r})}{|{\bf r}-{\bf R}_B|}d{\bf r} \\
&\leq -2\int_{\mathbb{R}^3}\frac{\varphi^2(|{\bf r}|)}{|{\bf r}|}d{\bf r}.
\end{split}
\end{equation}
Let $\rho_2({\bf r}_1,{\bf r}_2)=2|\psi({\bf r}_1)|^2\delta({\bf r}_2,{\bf r}_1-{\bf R}_A+{\bf R}_B)$,
we have
\begin{equation}\label{proof-4-e}
V_{\rm ee}^{\rm SCE}[2|\psi|^2] \leq
\int_{\mathbb{R}^6}\frac{\rho_2({\bf r}_1,{\bf r}_2)}{|{\bf r}_1-{\bf r}_2|}d{\bf r}_1d{\bf r}_2
= \int_{\mathbb{R}^3}\frac{2|\psi({\bf r}_1)|^2}{|{\bf R}_A-{\bf R}_B|}d{\bf r}_1 = \frac{1}{R}.
\end{equation}
Taking \eqref{proof-4-b}, \eqref{proof-4-c}, \eqref{proof-4-d} and \eqref{proof-4-e} into account, we have
\begin{equation}\label{proof-4-f}
E_{\rm SCE}(R) \leq 2e_0 + \frac{3}{2R}.
\end{equation}

To give a lower bound of $E_{\rm SCE}(R)$, we observe that
\begin{equation}\label{proof-4-g}
\begin{split}
E_{\rm SCE}(R) &\geq \frac{1}{2R} + \inf_{\|\phi\|_{L^2}=1}
\left\{\int_{\mathbb{R}^3}|\nabla\phi({\bf r})|^2d{\bf r}
+ 2\int_{\mathbb{R}^3}v_{\rm ext}({\bf r})|\phi({\bf r})|^2d{\bf r} \right\} \\
&= \frac{1}{2R} + 2\inf_{\|\phi\|_{L^2}=1}\langle\phi | \hat{h} | \phi\rangle,
\end{split}
\end{equation}
where $\hat{h}$ is the H$_2^+$ Hamiltonian $\hat{h}=-\frac{1}{2}\Delta + v_{\rm ext}$.
%
We claim that
\begin{equation}\label{proof-4-j}
\langle\phi |\hat{h} | \phi\rangle \geq (e_0-O(R^{-1}))\|\phi\|^2_{L^2}
\quad \forall~\phi\in H^1(\mathbb{R}^3).
\end{equation}
To show \eqref{proof-4-j}, we first decompose the unity function on $\mathbb{R}$ into
two smooth cutoff functions $\tilde{\zeta}_1$ and $\tilde{\zeta}_2$,
such that $\tilde{\zeta}_1^2+\tilde{\zeta}_2^2=1$,
$\tilde{\zeta}_1(x)=0$ for $x<-\frac{1}{2}$, $\tilde{\zeta}_2(x)=0$ for $x>\frac{1}{2}$,
and $|\nabla\tilde{\zeta}_i|\leq C^*$ with some constant $C^*$.
Let
$$
\zeta_i({\bf r})=\zeta_i(x,y,z)=\tilde{\zeta}_i(x/R) \quad i=1,2
$$
and $\phi_i=\zeta_i\phi$. We have $|\nabla\zeta_i|\leq C^*/R$ and
\begin{equation*}
\sum_{i=1}^2|\nabla\phi_i|^2 = \sum_{i=1}^2 (|\nabla\zeta_i|^2)|\phi|^2 + |\nabla\phi|^2.
\end{equation*}
Therefore,
\begin{equation*}
\begin{split}
\langle\phi |\hat{h} | \phi\rangle &=
\frac{1}{2}\sum_{i=1}^2\int_{\mathbb{R}^3}|\nabla\phi_i|^2 - \int_{\mathbb{R}^3}\sum_{i=1}^2 (|\nabla\zeta_i|^2) |\phi|^2
+ \int_{\mathbb{R}^3} v_{\rm ext}\big(|\phi_1|^2+|\phi_2|^2\big) \\
&\geq \sum_{i=1}^2\int_{\mathbb{R}^3}\left(\frac{1}{2}|\nabla\phi_i|^2 - \frac{|\phi_i|^2}{|{\bf r}|}\right)
- \left(\frac{2C^*}{R}\right)^2 - \frac{2}{R},
\end{split}
\end{equation*}
which implies \eqref{proof-4-j}.
Therefore, we obtain from \eqref{proof-4-g} and \eqref{proof-4-j} that
\begin{equation}\label{proof-4-h}
E_{\rm SCE}(R) \geq 2e_0 + O(R^{-1}).
\end{equation}
Together with the upper bound \eqref{proof-4-f}, this leads to \eqref{sce-h2-limit} at the limit $R\rightarrow\infty$.
\end{proof}

\begin{figure}[!ht]
\centering
\subfloat{\includegraphics[width=0.5\textwidth]{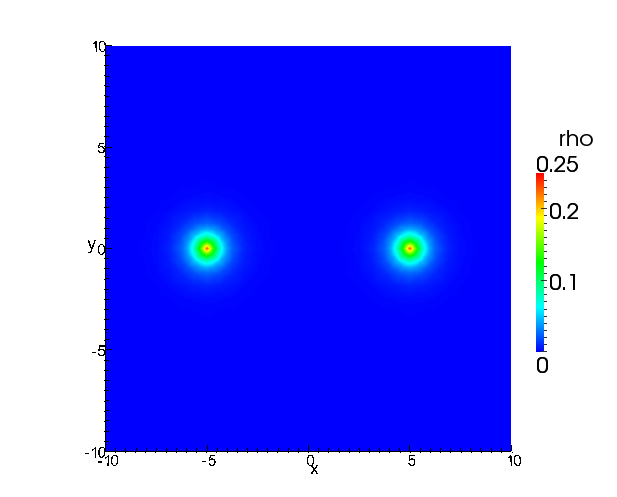}}
\subfloat{\includegraphics[width=0.5\textwidth]{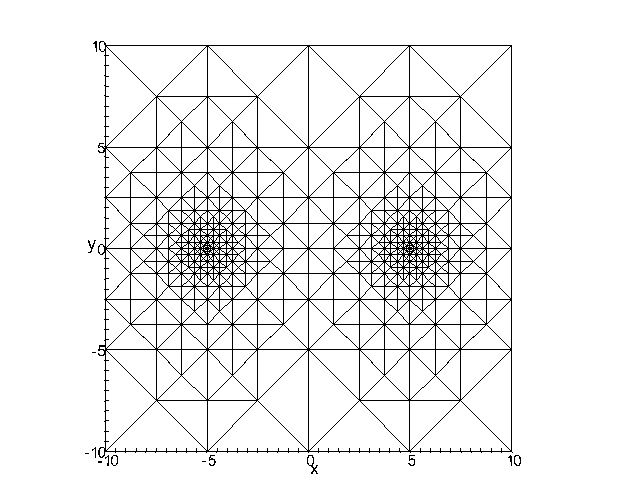}}
\caption{The electron density and corresponding mesh on slice $z=0$ ($R = 5$).}
\label{fig_h2_rho_mesh}
\end{figure}

In what follows, we present a numerical simulation of the dissociating H$_2$ molecule to support the theory.
The computations are carried out on a bounded domain $\Omega=[-10,10]^3$.
We use Algorithm~\ref{algorithm-ks_SCE} to solve \eqref{ks-sce} for the ground state energies
and electron densities for different bond length $2R$.
Concerning the optimal transport problem, we use the numerical methods introduced in
Section \ref{section-numerical} and calculate the co-motion function in each SCF iteration step.
The nonuniform mesh is generated by the package \textsf{PHG} \cite{phg},
a toolbox for parallel adaptive finite element programs developed at the State Key Laboratory of
Scientific and Engineering Computing of the Chinese Academy of Sciences. While the problem is effectively two-dimensional according to Theorem~\ref{theo-symmetric}, we have performed the calculations in three dimensions since \textsf{PHG} is tailored to three-dimensional problems. Figure~\ref{fig_h2_rho_mesh} shows a contour plot of
the electron density at slice $z=0$ and the corresponding mesh.
One observes that the grid reflects the higher density around the nuclei.

\begin{figure}[!ht]
\centering
\includegraphics[height=5.5cm]{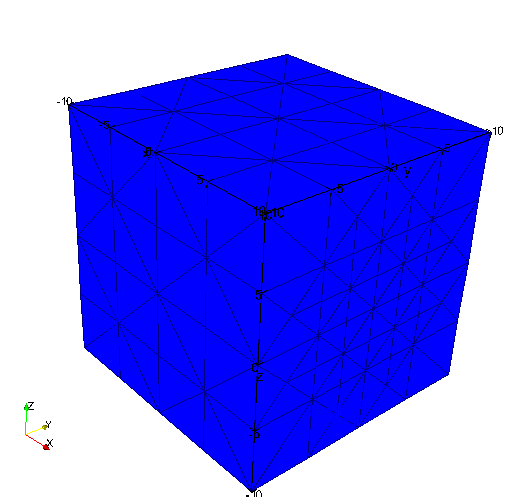}\hskip 3.0cm
\includegraphics[height=5.5cm]{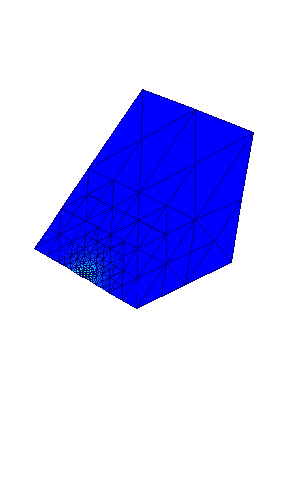}
\put(-140,70){\makebox(1,1){$x\geq 0,y\geq 0,z\geq 0,z\leq y$}}
\put(-120,60){\makebox(1,1){$\Longrightarrow$}}
\caption{Symmetric decomposition of the computation domain $\Omega$ for H$_2$.}
\label{fig_h2_mesh}
\end{figure}

To further reduce the computational cost, we can exploit the cylindrical symmetry of the system with the help of Theorem~\ref{theo-symmetric}. As shown in Figure~\ref{fig_h2_mesh}, the degrees of freedom can be reduced to $1/16$ of the original volume.
For the linear programming problem \eqref{ot_linprog}, we resort to \textsf{MOSEK} \cite{mosek}, a high-performance software for large-scale optimization problems.

The computational results are presented in Figure~\ref{fig_h2_bond}, in which we compare the bond energies in dependence of $R$ using the LDA and SCE Kohn-Sham methods, respectively. Here, the bond energies are the ground state energies of the systems minus $2e_0$, which is expected to be zero when the two hydrogen atoms are disassociated. Note that the SCE model shows the correct asymptotic behavior, while the LDA model fails at large $R$ by giving too large energies. For comparison, the LDA error $0.065\,{\rm a.u.}$ in \cite{ChallengesDFT2012} for the infinitely stretched H$_2$ molecule is lower than in Figure~\ref{fig_h2_bond} since twice the LDA-hydrogen energy is subtracted instead of twice the exact $e_0$, but remains significant; errors of similar magnitude are reported there for other functionals such as B3LYP or PBE.

\begin{figure}[!ht]
\centering
\includegraphics[width=0.7\textwidth]{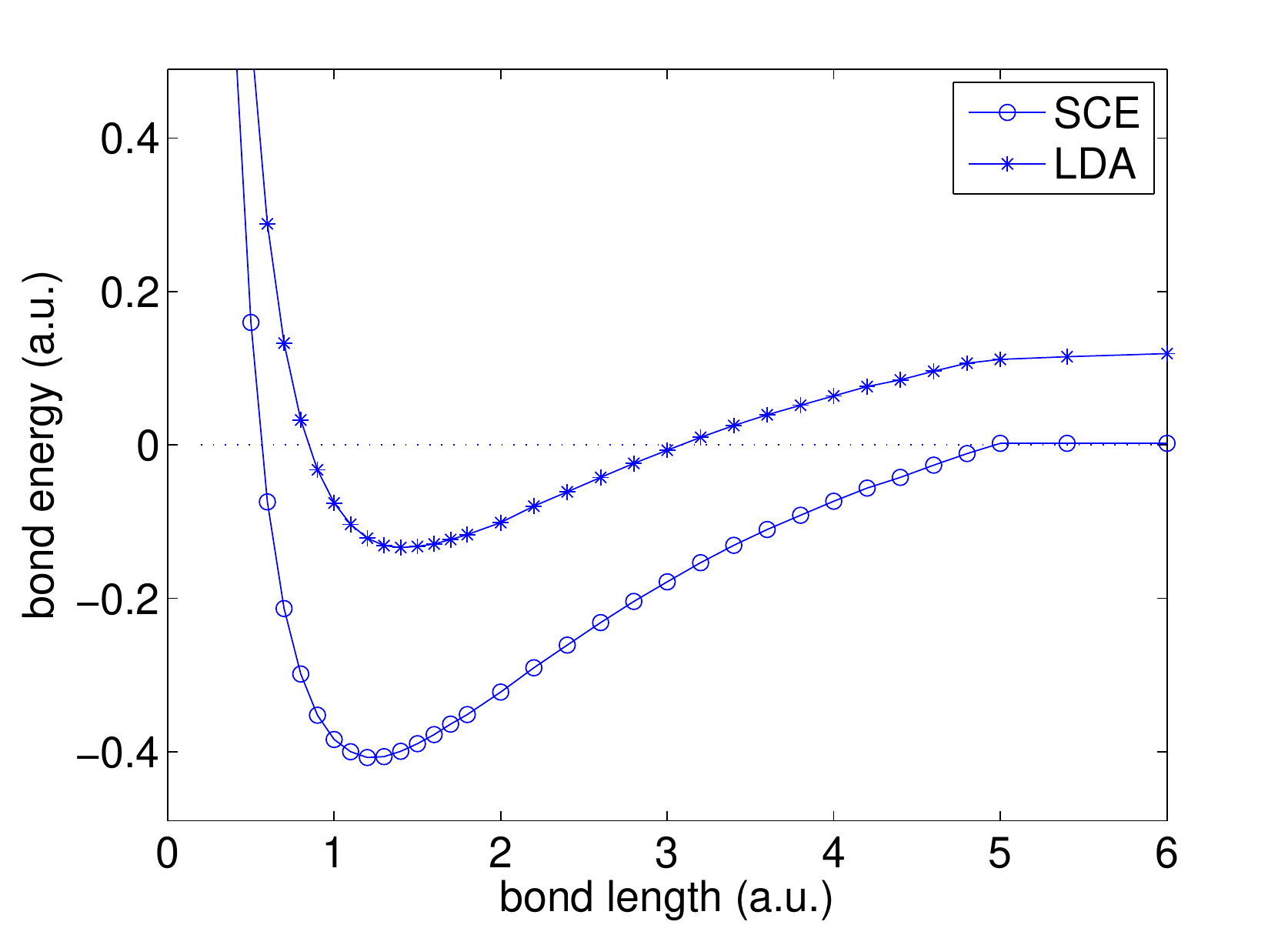}
\caption{H$_2$ potential energy curve as a function of the bond length for the SCE and LDA models.}
\label{fig_h2_bond}
\end{figure}

As physical explanations for these results, at long internuclear separations, if one electron is located near atom A, the other will be found close to atom B. This correlation is correctly reflected by the optimal transport model, hence the SCE model gives asymptotically the product of hydrogen orbitals on the two nuclei. In contrast, within the Kohn-Sham LDA framework, the two electrons are constrained to be in the same spatial orbital and each electron experiences only the average effect of the other, thus each electron has equal probability of being near A or B, irrespective of the position of the other electron. The possibility of both electrons being on the same atom is not excluded, as reflected in the wrong asymptotic behavior of the disassociation energy in Figure~\ref{fig_h2_bond}.

\begin{figure}[!htb]
\centering
\subfloat[]{\includegraphics[width=0.5\textwidth]{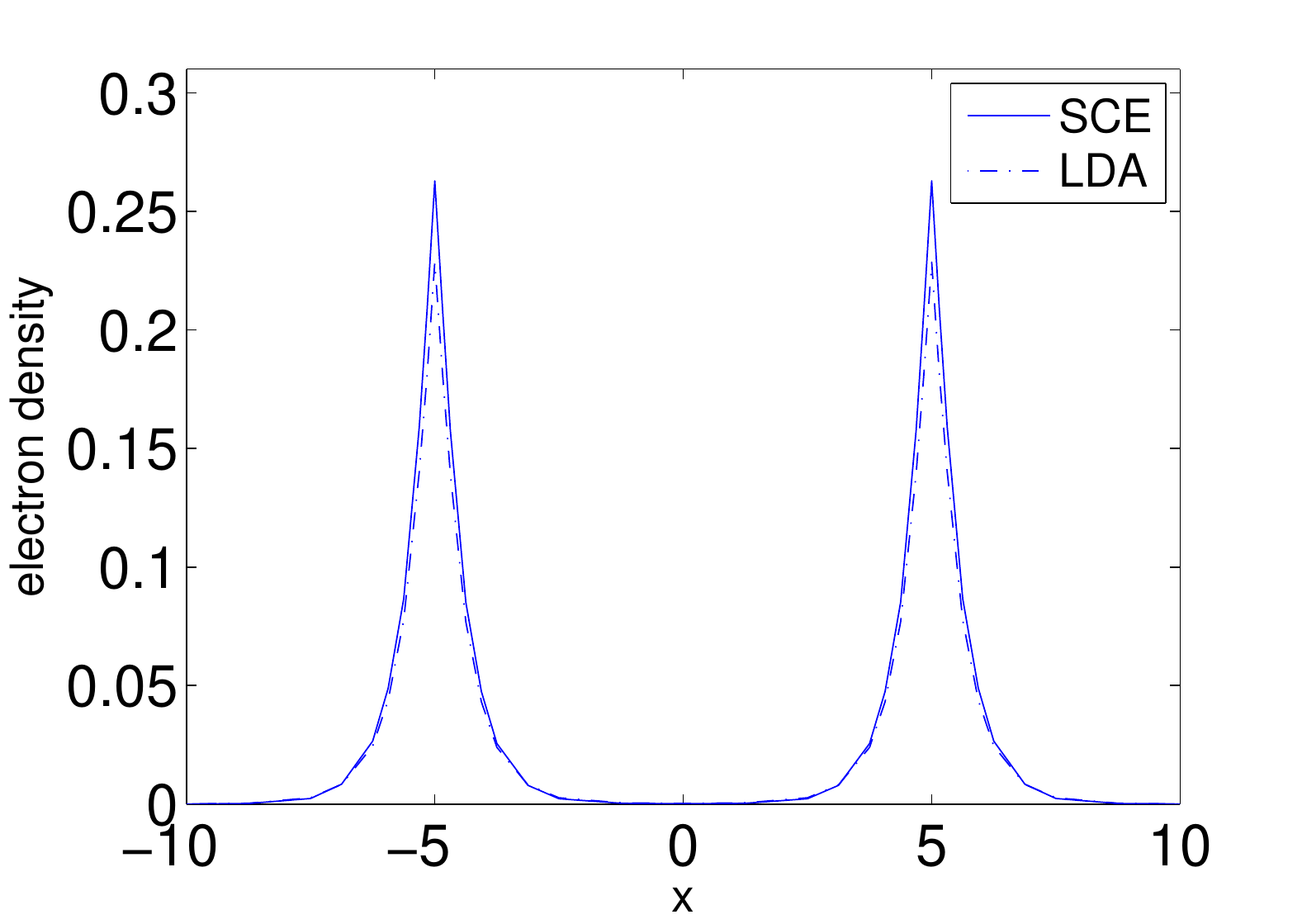}}
\subfloat[]{\includegraphics[width=0.5\textwidth]{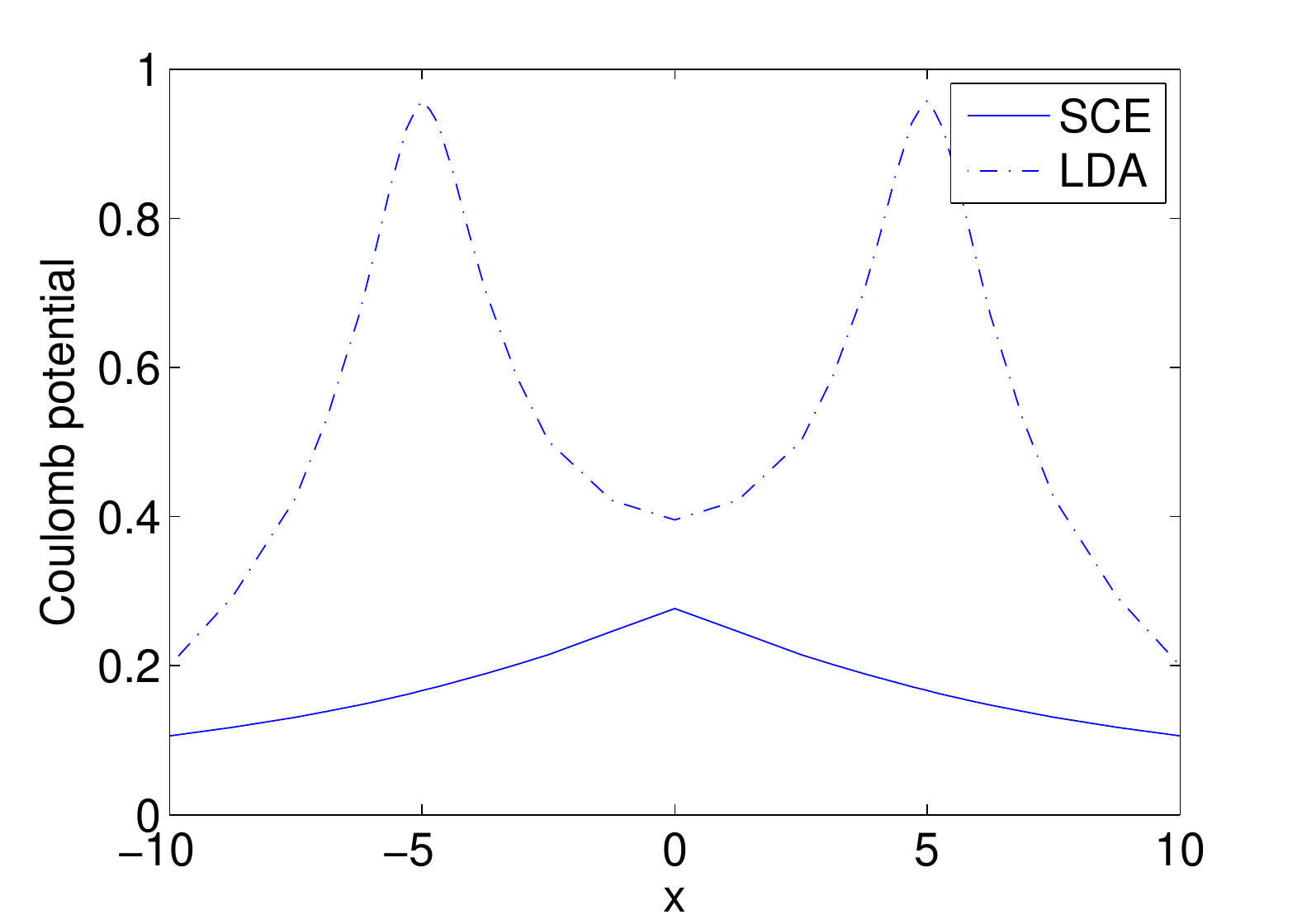}}
\caption{The SCE and LDA electron densities (a) and potentials (b) of H$_2$, plotted along the molecular axis.}
\label{fig_h2_x_rho_vcu}
\end{figure}

We further compare the electron densities and Coulomb potentials obtained by the two different models in Figure~\ref{fig_h2_x_rho_vcu}. Note that the scalar offset of the potential in the figure is determined by the boundary conditions, i.e., $\frac{N}{R}$ for LDA and $\frac{N-1}{R}$ for SCE \cite{seidl07}. The electron densities are actually quite close, while the shape of the potentials differs substantially when $R$ is large. The SCE potential is larger between the hydrogen atoms, favoring a depletion of the bond charge whenever the two atoms separate. This produces the correct results for large $R$.

\begin{figure}[!htb]
\centering
\includegraphics[width=0.4\textwidth]{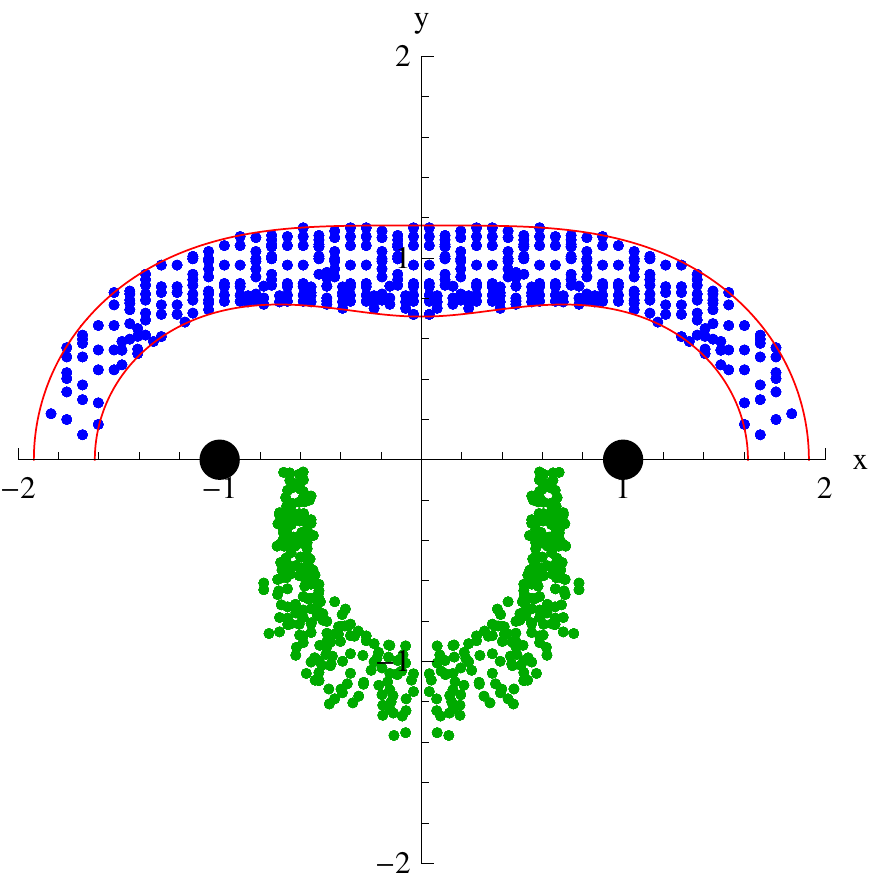}
\caption{Optimal transport mapping of the region $0.04 \le \rho({\bf r}) \le 0.08$ (indicated by the blue dots inside the red contours) to the green area, for the H$_2$ molecule with $R = 1$. Each blue dot corresponds to one barycenter in the numerical discretization, and has been rotated into the $x$-$y$-plane with $y \ge 0$ for visual clarity. The green dots are precisely the images of the blue dots under the optimal transport map.}
\label{fig_h2_ot_map}
\end{figure}

Finally, as illustration of the co-motion function or optimal map, Figure~\ref{fig_h2_ot_map} shows the image of the map on a density contour.

\section{Conclusions and perspectives}
\label{section-perspective}

The numerical discretization of the SCE optimal transport problem with Coulomb cost and two marginals leads to the linear programming problem \eqref{ot_linprog}, which can indeed be solved in practice, as we have demonstrated in a proof of concept calculation. The self-consistent SCE-DFT simulation of the dissociating H$_2$ molecule agrees well with the physically correct limit, unlike standard DFT models like LDA.

The theory of this paper applies to more general systems with arbitrary numbers of electrons,
however, the numerical algorithms need further developments. Specifically, we can restrict the $N$-particle density to the ansatz
\begin{equation*}
\rho_N({\bf r}_1,\dots,{\bf r}_N) = \frac{\rho({\bf r}_1)}{N}
\gamma_2({\bf r}_1,{\bf r}_2)\gamma_3({\bf r}_1,{\bf r}_3)
\cdots\gamma_N({\bf r}_1,\gamma_N),
\end{equation*}
where $\rho$ is the given single-particle density.
Here, $\gamma_j({\bf r}_1,{\bf r}_j)$ represents the probability of the $j$th electron being found at ${\bf r}_j$
while the first electron is located at ${\bf r}_1$.
We have
\begin{equation*}
\int_{\mathbb{R}^3}\gamma_j({\bf r}_1,{\bf r}_j)d{\bf r}_j = 1 \quad\quad j=2,\dots,N.
\end{equation*}
With a given discretization $\{e_k\}_{1\leq k\leq n}$ and barycenters ${\bf a}_k$ of $e_k$,
we can approximate $\gamma_j({\bf r}_1,{\bf r}_j)$ by a matrix $X_j=(x_{j,kl}) \in \mathbb{R}^{n\times n}$ for $j=2,\dots,N$.
Then the continuous model \eqref{energy-sce} is reduced to
\begin{equation}\label{ot_quadprog}
\begin{array}{rl}
\displaystyle \min_{X_2,\dots,X_N} &
\displaystyle \sum_{1 < j \leq N} \sum_{k,l=1}^n \frac{x_{j,kl}}{|{\bf a}_k - {\bf a}_l|}
\cdot \frac{\rho_k}{N}
+ \sum_{1 < i < j\leq N} \sum_{k,l,l' = 1}^n \frac{x_{i,kl} \cdot x_{j,kl'}}{|{\bf a}_l - {\bf a}_{l'}|} 
\cdot \frac{\rho_k}{N} \\ \\
\mathrm{s.t.} & \sum_{k=1}^n x_{i,kl}=1, \quad l = 1, \dots, n, \quad i = 2, \dots, N \\[1ex]
& \sum_{l=1}^n x_{i,kl} = 1, \quad k = 1, \dots, n, \quad i = 2, \dots, N \\[1ex]
& x_{i,kl} \ge 0.
\end{array}
\end{equation}
This is a quadratic programming problem of the form
\begin{align*}
&\min_x \, x^T H x + f^T x \\
&\mathrm{s.t.}\quad A\,x = b \quad\mathrm{and}\quad x_k \ge 0.
\end{align*}
By solving the above quadratic programming problem, we approximate the co-motion functions
$T_2,\dots,T_N$ by the matrices $X_2=(x_{2,kl}), \dots, X_N = (x_{N,kl})$.
Similar to \eqref{T_approximate}, the co-motion functions can be approximated by
$$
T_i({\bf a}_k) \approx \sum_{1\leq l\leq n} {\bf a}_l \, x_{i,kl}
\qquad k=1,\dots,n, \quad i=2,\dots,N.
$$
However, a serious difficulty in solving \eqref{ot_quadprog} stems from the non-convexity of the matrix $H$. Moreover, the symmetric decomposition is not clear for systems with more than two electrons. We plan to investigate these issues in future work.

\medskip

{\bf Acknowledgments}\quad C.M.\ acknowledges support from the DFG project FR~1275/3-1.

\appendix

\section{Appendix}

\begin{lemma}\label{lemma-outerSGS}
If $v_0$ is a maximizer of \eqref{SGS} with single-particle density $\rho$,
and $\rho^0_N$ is a minimizer of the inner optimization of \eqref{SGS}, i.e. the minimizer of
\begin{equation}\label{outerSGS}
\min_{\rho_N}\int_{\mathbb{R}^{3N}}\left(c_{\rm ee}({\bf r}_1,\dots,{\bf r}_N)
-\sum_{i=1}^N v_0({\bf r}_i)\right)\rho_N({\bf r}_1,\dots,{\bf r}_N)
d{\bf r}_1\dots d{\bf r}_N,
\end{equation}
then $\rho_N^0$ is exactly the minimizer of the constraint minimization problem \eqref{energy-sce}.
\end{lemma}

\begin{proof}
Since $v_0$ is a maximizer of \eqref{SGS}, we have
\begin{equation}\label{proof-5-a}
0=\left.\frac{d}{d\varepsilon}\right|_{\varepsilon=0}
\left( \int_{\mathbb{R}^3}(v_0+\varepsilon\tilde{v})\rho
+ \min_{\rho_N}\int_{\mathbb{R}^{3N}}\big(c_{\rm ee}-V_0-\varepsilon\tilde{V}\big)\rho_N \right)
\end{equation}
for any $\tilde{v}$,
where $V_0({\bf r}_1,\dots,{\bf r}_N)=\sum_{i=1}^N v_0({\bf r}_i)$ and
$\tilde{V}({\bf r}_1,\dots,{\bf r}_N)=\sum_{i=1}^N \tilde{v}({\bf r}_i)$.

We view the minimization over $\rho_N$ in \eqref{proof-5-a} as a quadratic variational problem
for (spinless bosonic) normalized wavefunctions $\Phi$ (by identifying $|\Phi|^2=\rho_N$):
\begin{equation*}
\min_{\rho_N}\int_{\mathbb{R}^{3N}}(c_{\rm ee}-V_0-\varepsilon\tilde{V})\rho_N
= \min_{\Phi}\langle \Phi | c_{\rm ee}-V_0-\varepsilon\tilde{V} | \Phi \rangle.
\end{equation*}
Using first order perturbation theory and the fact that
$\rho_N^0$ is the minimizer of $\min_{\rho_N}\int_{\mathbb{R}^{3N}}(c_{\rm ee}-V_0)\rho_N$, we obtain
\begin{equation}\label{proof-5-b}
\min_{\rho_N}\int_{\mathbb{R}^{3N}}\left(c_{\rm ee}-V_0-\varepsilon\tilde{V}\right)\rho_N
=\int_{\mathbb{R}^{3N}}(c_{\rm ee}-V_0)\rho_N^0 + \varepsilon\int_{\mathbb{R}^{3N}}\tilde{V}\rho_N^0
+ O(\varepsilon^2).~
\end{equation}
Substituting \eqref{proof-5-b} into \eqref{proof-5-a} yields
\begin{equation*}
\begin{split}
0 &= \int_{\mathbb{R}^{3N}}\tilde{V}\rho_N^0-\int_{\mathbb{R}^3}\tilde{v}\rho \\
&= -\int_{\mathbb{R}^3}\tilde{v}({\bf r})\left(\rho({\bf r})-N\int_{\mathbb{R}^{3(N-1)}}
\rho_N^0({\bf r},{\bf r}_2,\dots,{\bf r}_N)d{\bf r}_2\cdots d{\bf r}_N\right)d{\bf r},
\end{split}
\end{equation*}
which indicates $\rho_N^0\mapsto\rho$.
That is, if $v_0$ is a maximizer of the outer optimization of \eqref{SGS},
then the minimizer of the associated inner optimization in \eqref{SGS}
automatically has the single-particle density $\rho$.

Moreover, since the term $\int_{\mathbb{R}^{3N}}V_0\rho_N$ only depends on the single-particle density
of $\rho_N$, $\rho_N^0$ must minimize $\int_{\mathbb{R}^{3N}}c_{\rm ee}\rho_N$ under the constraint
$\rho_N\mapsto\rho$ (because any other minimizer of \eqref{energy-sce} gives the same value for
$\int_{\mathbb{R}^{3N}}(c_{\rm ee}-V)\rho_N$ as $\rho_N^0$).
Therefore, $\rho_N^0$ is also a minimizer of \eqref{energy-sce}, which completes the proof.
\end{proof}

{\small

}

\end{document}